\documentclass{article}
\usepackage[margin=1in]{geometry}  

\usepackage{listings}
\usepackage{amsmath,mathtools}
\usepackage{amsthm}
\usepackage{tikz}
\usepackage{caption,subcaption}
\usepackage{array}
\usepackage{mdwmath}
\usepackage{multirow}
\usepackage{mdwtab}
\usepackage{eqparbox}
\usepackage{multicol}
\usepackage{amsfonts}
\usepackage{multirow,bigstrut,threeparttable}
\usepackage{amsthm}
\usepackage{array}
\usepackage{bbm}
\usepackage{subcaption}
\usepackage{epstopdf}
\usepackage{mdwmath}
\usepackage{mdwtab}
\usepackage{eqparbox}
\usepackage{latexsym}
\usepackage{amssymb}
\usepackage{bm}
\usepackage{amssymb}
\usepackage{graphicx}
\usepackage{mathrsfs}
\usepackage{epsfig}
\usepackage{psfrag}
\usepackage{setspace}
\usepackage{tikz-cd}

\usepackage[
CJKbookmarks=true,
bookmarksnumbered=true,
bookmarksopen=true,
colorlinks=true,
citecolor=red,
linkcolor=blue,
anchorcolor=red,
urlcolor=blue
]{hyperref}
\usepackage{cleveref}
\usepackage{algorithm}
\usepackage{algorithmic}
\usepackage{stfloats}
\usepackage[
autocite    = superscript,
backend     = bibtex, 
sortcites   = true,
style       = authoryear 
]{biblatex} 
\usepackage{nameref}
\usepackage[normalem]{ulem}

\usepackage{soul}

\input xy
\xyoption{all}

\theoremstyle{plain}

\newtheorem{proposition}{Proposition}
\newtheorem{lemma}{Lemma}[section]

\crefname{theorem}{Theorem}{Theorems}
\crefname{proposition}{Proposition}{Propositions}
\crefname{corollary}{Corollary}{Corollaries}

\theoremstyle{definition}

\newtheorem{remark}{Remark}

\crefname{definition}{Definition}{Definitions}
\crefname{remark}{Remark}{Remarks}

\def\EE{\mathbb{E}}

\def\NN{\mathbb{N}}

\def\PP{\mathbb{P}}

\def\calA{\mathcal{A}}
\def\calB{\mathcal{B}}

\def\calD{\mathcal{D}}

\def\calF{\mathcal{F}}
\def\calG{\mathcal{G}}
\def\calH{\mathcal{H}}
\def\calI{\mathcal{I}}

\def\calN{\mathcal{N}}

\def\1{\mathbbm{1}}

\def\independenT#1#2{\mathrel{\rlap{$#1#2$}\mkern2mu{#1#2}}}

\usepackage{booktabs}
\usepackage{adjustbox}
\usepackage{multirow}
\usepackage{listings}

\theoremstyle{plain}

\usepackage{xspace}

\newcommand\indep{\protect\mathpalette{\protect\independenT}{\perp}}
\def\independenT#1#2{\mathrel{\rlap{$#1#2$}\mkern2mu{#1#2}}}

\definecolor{myblue}{rgb}{.8, .8, 1}
\definecolor{mathblue}{rgb}{0.2472, 0.24, 0.6} 
\definecolor{mathred}{rgb}{0.6, 0.24, 0.442893}
\definecolor{mathyellow}{rgb}{0.6, 0.547014, 0.24}

\usepackage{enumitem}

\newcommand{\No}{{n}}
\newcommand{\NoNull}{{n_0}}

\newcommand{\NoNc}{m}
\newcommand{\probNull}{{\pi_0}}

\newcommand{\pval}[1]{{p_{#1}}}

\newcommand{\testStatistics}[1]{{T_{#1}}}

\newcommand{\cdfTestStatistics}[1]{{F_{#1}}}
\newcommand{\cdfTestStatisticsNull}{{F_{0}}}
\newcommand{\cdfTestStatisticsNonNull}{{F_{1}}}

\newcommand{\hypothesis}[1]{{H_{#1}}}

\newcommand{\hypothesisIndex}[1]{{\calI_{#1}}}

\title{Adaptive Storey's null proportion estimator}

\author{Zijun Gao\thanks{Marshall Business School, University of Southern California, CA, USA. Email: \{zijungao\}@marshall.usc.edu.}}

\begin{document}

\maketitle

\begin{abstract}
    False discovery rate (FDR) is a commonly used criterion in multiple testing and the Benjamini-Hochberg (BH) procedure is arguably the most popular approach with FDR guarantee.
    To improve power, the adaptive BH procedure has been proposed by incorporating various null proportion estimators, among which Storey's estimator has gained substantial popularity.
    The performance of Storey's estimator hinges on a critical hyper-parameter, where a pre-fixed configuration lacks power and existing data-driven hyper-parameters compromise the FDR control.
    In this work, we propose a novel class of adaptive hyper-parameters and establish the FDR control of the associated BH procedure using a martingale argument.
    Within this class of data-driven hyper-parameters, we present a specific configuration designed to maximize the number of rejections and characterize the convergence of this proposal to the optimal hyper-parameter under a commonly-used mixture model.
    We evaluate our adaptive Storey's null proportion estimator and the associated BH procedure on extensive simulated data and a motivating protein dataset. 
    Our proposal exhibits significant power gains when dealing with a considerable proportion of weak non-nulls or a conservative null distribution.
\end{abstract}

\section{Introduction}\label{sec:introduction}

The rapid advancement of data collection technologies, including high-throughput experiments, has sparked strong interest in the multiple testing problem over the past a few decades and led to various remarkable statistical advancements.
To fix notations, we suppose there are $\No$ hypotheses $H_i$ associated with p-values $p_i \in (0,1)$, and $p_i \sim U(0,1)$ for $H_i = 0$.
Provided with a multiple testing method, we use $V$ and $R$ to represent the number of falsely rejected null hypotheses (false discoveries) and the total number of rejections, respectively. 
In this paper, we focus on the control of false discovery rate (FDR) \parencite{benjamini1995controlling} defined as the expected proportion of false discoveries $\EE[V/(R \vee 1)]$.

The Benjamini-Hochberg (BH) procedure \parencite{benjamini1995controlling} is a widely-used multiple testing method to control FDR.
The BH procedure applied to independent p-values at level $q$ effectively controls the FDR at a lower level $\probNull q$.
Here $\NoNull$ denotes the number of null hypotheses and $\probNull = \NoNull / \No$ denotes the proportion of nulls.
The BH procedure might be overly conservative when $\probNull$ is small, a point reiterated in various papers, including \cite{benjamini2006adaptive, finner_controlling_2009}.

To mitigate the conservativeness, the adaptive BH has been proposed, which takes an estimator of the null proportion $\hat{\pi}_0$ and applies the standard BH procedure at the level $q/\hat{\pi}_0$, typically higher than $q$.
There are several null proportion estimators\footnote{We remark that estimating $\pi_0$, also under the name prevalence estimation, is not limited to the realm of FDR control. We refer readers to \cite[Chapter 8]{cui_handbook_2021} for an in-depth review of this topic.} whose associated adaptive BH procedure properly controls the FDR.
\cite{benjamini2006adaptive} apply the BH procedure at a reduced level $q / (1 + q)$ in step one, and estimate the null proportion by $(1+q)(1 - R_1/\No)$ in step two, where $R_1$ denotes the number of rejections from the first step.
In situations where the signals are weak, it's possible that $R_1$ is close to zero and the resulting null proportion estimator exceeds one, leading to the associated adaptive BH procedure being more conservative than the standard counterpart.
Another popular null estimator, which we focus on in this paper, takes the form
\begin{align}\label{eq:Storey.null.prop}
    \hat{\pi}_{0}^{\lambda} = \frac{1 + \sum_{i=1}^n \1_{\{p_i \ge \lambda\}}}{n(1 - \lambda)}.
\end{align}
Here $\lambda \in (0,1)$ is a pre-fixed hyper-parameter.
The estimator appeared in \cite{schweder_plots_1982}, acquired the unofficial name ``Storey's estimator'' since \cite{storey2002direct}, and has gained significant recognition after  \cite{storey2004strong} showing the adaptive BH procedure with $ \hat{\pi}_{0}^{\lambda}$ using a pre-fixed $\lambda$ controls FDR.

The performance of Storey's estimator and the associated adaptive BH procedure may depend heavily on the choice of $\lambda$, which we illustrate using a motivating protein study.
The protein dataset is collected in a study aiming to identify cell membrane proteins that are labeled with biotin by horseradish peroxidase (HRP).
Before measuring the abundance of proteins by mass spectrometry, a process is executed to capture the biotinylated proteins that are likely to be enriched by streptavidin bead and to filter out the majority of background proteins (proteins that are not affected by HRP) \parencite{shuster2022situ}.
\Cref{fig:protein.hist.intro} (a) shows the histogram of the p-values derived from the protein abundances: there appear to be a few strong non-nulls with p-values very close to zero; a substantial fraction of weak non-nulls with moderately small p-values; a proportion of conservative nulls corresponding to p-values larger than $0.8$ which stochastically dominate $U(0.8,1)$.
As the hyper-parameter $\lambda$ grows from $0.2$ to $1$, Storey's estimator $\hat{\pi}_0^{\lambda}$ first decreases then increases and eventually exceeds one (\Cref{fig:protein.hist.intro} (b)). 
Meanwhile, the number of rejections produced by the associated adaptive BH procedure at the target FDR level $0.2$ increases from $16$ to over $30$, and subsequently falls back to below $16$ (\Cref{fig:protein.hist.intro} (c)).

Multiple ways have been suggested for choosing $\lambda$.
For pre-fixed $\lambda$, the configuration $\lambda = 1/2$ is employed by \cite{storey2003statistical}, while \cite{blanchard2009adaptive} suggest using a smaller $\lambda = q$. 
For data-driven options, \cite{storey2002direct} suggests to tune $\lambda$ by bootstrap; \cite{storey2003statistical} discuss applying spline smoothing to $\hat{\pi}_0^{\lambda}$ computed with a sequence of $\lambda$ and evaluating the fit at $\lambda = 1$.
Nevertheless, empirical evidence \parencite{black_note_2004, jiang_estimating_2008} shows that the data-driven null proportion estimators may be downward-biased and the associated adaptive BH procedure may not control FDR.

In this paper, we propose a class of data-driven $\hat{\lambda}$ with the desirable FDR guarantee.
In particular, we introduce a filtration and let $\hat{\lambda}$ be arbitrary stopping time with regard to the filtration. 
We prove the associated adaptive BH procedure controls the FDR by constructing a martingale adapted to the filtration and employing the optional stopping theorem.
Among this class of data-driven hyper-parameters, we provide a concrete stopping time designed to yield the maximal number of rejections, and establish its convergence to the optimal $\lambda$ maximizing the super-population power under a mixture model.
On the motivating protein dataset, a pre-fixed $\lambda = q = 0.2$ gives $\hat{\pi}_0^{0.2} = 0.85$, and the corresponding BH yields $16$ rejections; our method adaptively chooses
$\hat{\lambda} = 0.68$, leading to $\hat{\pi}_0^{0.68} = 0.62$ and $31$ rejections.
In addition, extensive simulations have also demonstrated the proposed data-driven $\hat{\lambda}$ consistently results in more rejections than a pre-fixed $\lambda$ across various data generation mechanisms.

Below, we describe the class of data-driven $\hat{\lambda}$
in \Cref{sec:method} and prove the finite-sample FDR control in \Cref{sec:theory}.
We delve into a specific stopping time devised to maximize the power in \Cref{sec:stopping.time}. 
\Cref{sec:simulations}, \Cref{sec:real.data} contain the simulation studies, the analysis of the motivating protein dataset, respectively.
We conclude with discussions on future directions in \Cref{sec:discussion}.

\begin{figure}[tbp]
        \centering
        \begin{minipage}{0.3\textwidth}
                \centering
                \includegraphics[clip, trim = 0cm 0cm 0cm 0cm, height = 5cm]{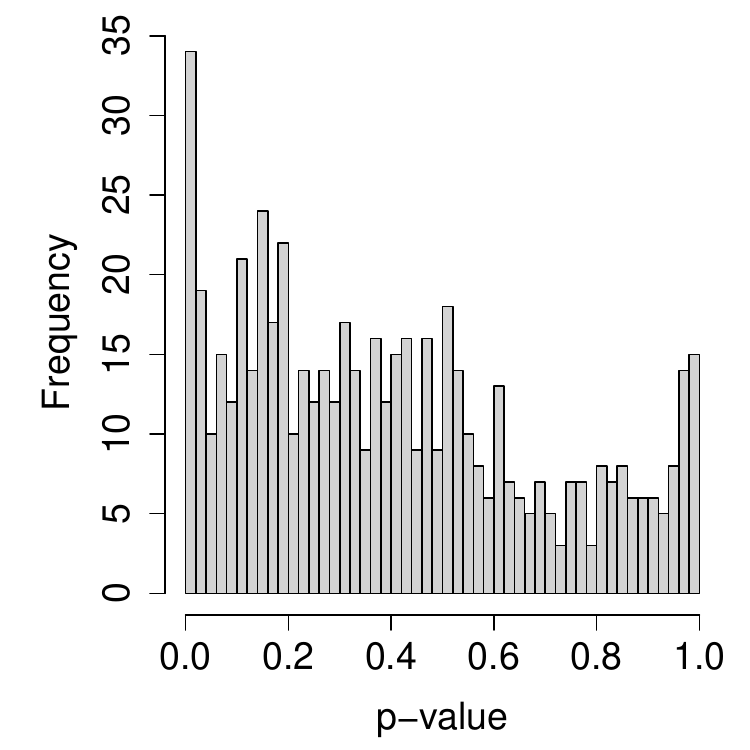}
                \subcaption*{\qquad (a)}
        \end{minipage}
        \begin{minipage}{0.3\textwidth}
            \centering
            \includegraphics[clip, trim = 0cm 0cm 0cm 0cm, height = 5cm]{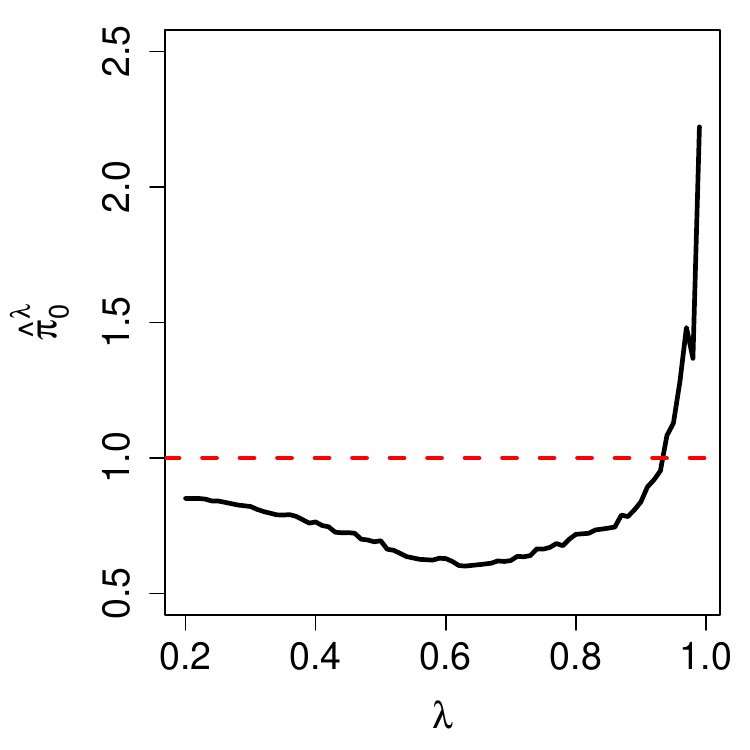}
                \subcaption*{\qquad \quad (b)}
        \end{minipage}
        \begin{minipage}{0.3\textwidth}
                \centering
                \includegraphics[clip, trim = 0cm 0cm 0cm 0cm, height = 5cm]{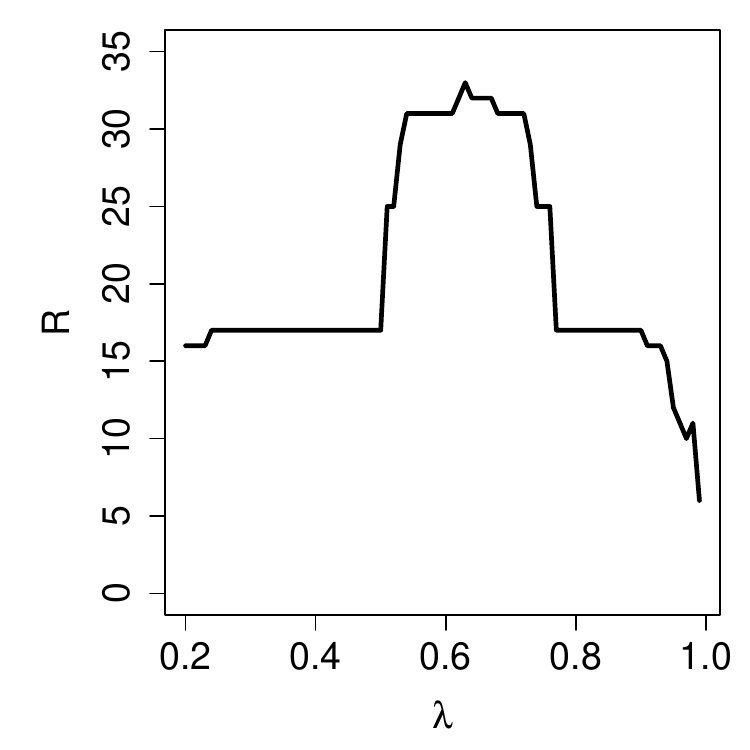}
                \subcaption*{\qquad \quad (c)}
        \end{minipage}
            \caption{
            Motivating protein dataset. 
            Panel (a) plots the p-value histogram of $585$ proteins.
            Panel (b) plots Storey's estimator~\eqref{eq:Storey.null.prop} across $\lambda$, and panel (c) shows the number of rejections of the associated adaptive BH procedure.
            }
        \label{fig:protein.hist.intro}
\end{figure}

\section{A class of adaptive Storey's null proportion estimators}\label{sec:adaptive.Storey}


\subsection{Stopping time $\hat{\lambda}$}\label{sec:method}

Our data-driven hyper-parameter hinges on the filtration defined as,
\begin{align*}
    \calG_{t} := \sigma\left(\sum_{i \in [n]} \1_{\{p_i \ge s\}}: q \le s \le t\right), \quad q \le t \le 1.
\end{align*}
The filtration is increasing and only depends on the p-values larger than $q$.
The increasing filtration is different from the decreasing filtration $\calH_{t} = \sigma\left(\sum_{i \in [n]} \1_{\{p_i \ge s\}}: s \ge t\right)$ used by \cite{storey2004strong} to show the BH procedure controls FDR.

We propose $\hat{\lambda}$ to be any stopping time with regard to $\calG_t$.
For instance, $\hat{\lambda}$ can be a pre-fixed value independent of the data, $\inf\{t: \hat{\pi}_0^{t} \le 0.8 \hat{\pi}_0^{q}\}$, $\inf\{t: \hat{\pi}_0^{t} \le \hat{\pi}_0^{q} - 0.2\}$. In \Cref{sec:stopping.time}, we describe a less ad-hoc stopping time with the aim to maximize the number of rejections.
Provided with a stopping time $\hat{\lambda}$, we further compute $\hat{\pi}_{0}^{\hat{\lambda}}$, referred to as adaptive Storey (AS) estimator below, by plugging $\hat{\lambda}$ in Eq.~\eqref{eq:Storey.null.prop}.
Finally, we combine $\hat{\pi}_{0}^{\hat{\lambda}}$ and the BH procedure as described in \Cref{algo:BH.adaptive.Storey}.

\Cref{fig:BH.adaptive.Storey} provides an illustration of \Cref{algo:BH.adaptive.Storey}:
in step one, starting from the FDR level $q$, we gradually increase the hyper-parameter $\lambda$, along which the p-values $p_i \ge q$ are sequentially disclosed in ascending order. Once a pre-defined stopping rule is satisfied, we stop and output the hyper-parameter value $\hat{\lambda}$;
in step two, also initiated at $q$, we decrease the threshold $\tau$ of the BH procedure, and inspect $p_i < q$ in descending order. The process is terminated when the estimated FDR augmented by $\hat{\pi}_0^{\hat{\lambda}}$, i.e., $\widehat{\text{FDP}}(t; \hat{\pi}_0^{\hat{\lambda}}) := \hat{\pi}_0^{\hat{\lambda}} \No t/{\sum_{i=1}^{\No}\1_{\{p_{i} \le t\}}}$, falls below the target level $q$.  The value at which we stop is reported as the rejection threshold of the adaptive BH procedure.

\begin{figure}[h]
    \centering
    \begin{tikzpicture}
    \centering
    \draw[->] (-0.5,0) -- (10.5,0);
    \draw (0,0.1) -- (0,-0.1) node[below] {$0$};
    \draw[dashed] (5,1.2) -- (5,-0.2) node[below] {$q$};
    \draw (8,0.1) -- (8,-0.1) node[below] {$\hat{\lambda}$};
    \draw (10,0.1) -- (10,-0.1) node[below] {$1$};
    \draw (0.25,0.1) -- (0.25,-0.1) node[above] {$p_{(1)}$};
    \draw (0.85,0.1) -- (0.85,-0.1) node[above] {$p_{(2)}$};
    \draw (1.5,0.1) -- (1.5,-0.1) node[above] {$p_{(3)}$};
    \draw (4,0.1) -- (4,-0.1) node[above] {$p_{(4)}$};
    \draw (5.5,0.1) -- (5.5,-0.1) node[above] {$p_{(5)}$};
    \draw (6.5,0.1) -- (6.5,-0.1) node[above] {$p_{(6)}$};
    \draw (7.25,0.1) -- (7.25,-0.1) node[above] {$p_{(7)}$};
    \draw (9,0.1) -- (9,-0.1) node[above] {$p_{(8)}$};
    
    \draw[->] (5.1,0.75) -- (8,0.75) node[midway, above] {AS};
    \node at (-2,0) {Step 1};
    \end{tikzpicture}

    \begin{tikzpicture}
    \centering
    \draw[->] (-0.5,0) -- (10.5,0);
    \draw (0,0.1) -- (0,-0.1) node[below] {$0$};
    \draw (2,0.1) -- (2,-0.1) node[below] {$\hat{\tau}$};
    \draw[dashed] (5,1.2) -- (5,-0.2) node[below] {$q$};
    \draw (8,0.1) -- (8,-0.1) node[below] {$\hat{\lambda}$};
    \draw (10,0.1) -- (10,-0.1) node[below] {$1$};
    \draw (0.25,0.1) -- (0.25,-0.1) node[above] {$p_{(1)}$};
    \draw (0.85,0.1) -- (0.85,-0.1) node[above] {$p_{(2)}$};
    \draw (1.5,0.1) -- (1.5,-0.1) node[above] {$p_{(3)}$};
    \draw (4,0.1) -- (4,-0.1) node[above] {$p_{(4)}$};
    \draw (5.5,0.1) -- (5.5,-0.1) node[above] {$p_{(5)}$};
    \draw (6.5,0.1) -- (6.5,-0.1) node[above] {$p_{(6)}$};
    \draw (7.25,0.1) -- (7.25,-0.1) node[above] {$p_{(7)}$};
    \draw (9,0.1) -- (9,-0.1) node[above] {$p_{(8)}$};
    
    \draw[->] (4.9,0.75) -- (2,0.75) node[midway, above] {BH with $\hat{\pi}_0^{\hat{\lambda}}$};
    \node at (-2,0) {Step 2};
    \end{tikzpicture}
    \caption{The step-up adaptive BH procedure combined with the step-down AS null proportion estimator.} 
    \label{fig:BH.adaptive.Storey} 
\end{figure}
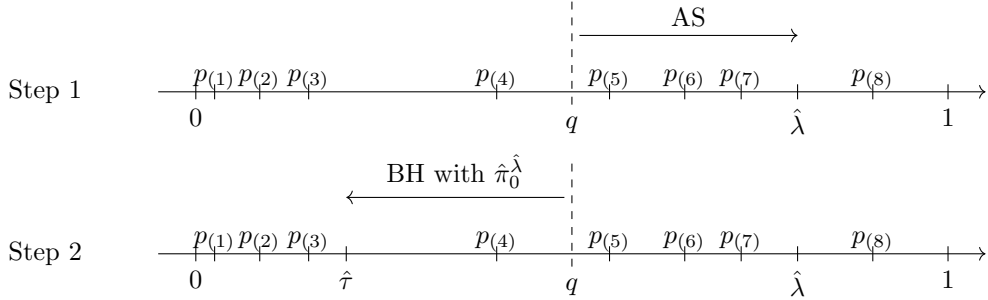

\begin{algorithm}\caption{Adaptive BH with the AS null proportion estimator}\label{algo:BH.adaptive.Storey}
    \begin{algorithmic}
        \STATE \textbf{Input}: p-values $p_i$, $i \in [n]$, FDR level $q$, a stopping rule adapted to $\calG_t$.
         \STATE 1. Compute the stopping time $\hat{\lambda}$ based on the stopping rule. 
         Calculate $\hat{\pi}_{0}^{\hat{\lambda}}$ as in Eq.~\eqref{eq:Storey.null.prop}.
         \STATE 2. Compute the BH rejection threshold
         \begin{align*}
             \hat{\tau} = \sup \left\{t < q: \frac{\hat{\pi}_0^{\hat{\lambda}}nt}{\sum_{i=1}^n \1_{\{p_{i} \le t\}}} \le q\right\}.
         \end{align*}
         Here we regard $\sup \emptyset = -\infty$.\\
         \STATE \textbf{Output}: rejection set $\{H_i: p_i \le \hat{\tau}\}$.
    \end{algorithmic}
\end{algorithm}

\subsection{Theoretical FDR guarantee}\label{sec:theory}

\subsubsection{Independent p-values}

The following result demonstrates that the BH procedure combined with the AS null proportion estimator controls the FDR when applied to independent p-values.

\begin{proposition}\label{prop:BH.adaptive.Storey}
    Let $\hat{\lambda}$ be a stopping time with regard to $\calG_t$. 
    Assume for all $i \in [\No]$ such that $H_i = 0$, we have $p_i \sim U(0,1)$ and $p_i \indep \bm (p_j)_{j \neq i}$.
    Then \Cref{algo:BH.adaptive.Storey} controls the FDR at the target level.
\end{proposition}

We provide an outline of the proof here and include the details in the appendix.
Define an extended increasing filtration $\calF_{t} := \sigma\left( \sum_{H_i = 0} \1_{\{p_i \ge s\}}, \sum_{i \in [n]} \1_{\{p_i \ge s\}}: q \le s \le t\right)$, $q \le t \le 1$, which satisfies $\calG_t \subseteq \calF_t$.
Any stopping time with regard to $\calG_t$ is also a stopping time regarding $\calF_t$.
We define a random process
\begin{align}\label{eq:martingale} 
    M_t := \frac{1 - t}{1 + \sum_{H_i = 0} \1_{\{p_i \ge t\}}} \ge \frac{1 - t}{1 + \sum_{i=1}^\No \1_{\{p_i \ge t\}}} = \frac{1}{\No \hat{\pi}_0^{\hat{\lambda}}}, \quad q \le t \le 1.
\end{align}
A key observation is that $M_t$ is a non-negative super-martingale with respect to $\calF_{t}$. 
Then by the optional stopping time theorem,
\begin{align}\label{eq:expectation.inverse}
    \EE\left[\frac{1}{\hat{\pi}_0^{\hat{\lambda}}} \right]
    \le \EE\left[\No M_{\hat{\lambda}}\right]
     \le \EE\left[\No M_{q} \right]
     \le \frac{\No}{1 + \NoNull},
\end{align}
where the last inequality comes from a standard result of binomial distribution.
Finally, we combine \eqref{eq:expectation.inverse} and a leave-one-out argument adapted from \cite{benjamini2014selective} to prove the finite-sample FDR control of \Cref{algo:BH.adaptive.Storey}.

\begin{remark}\label{rmk:stochastic.dominance}
    \Cref{prop:BH.adaptive.Storey} holds for non-uniform null p-values whose marginal distribution satisfies the conditional stochastic dominance property
    \begin{align}\label{eq:conditional.stochastic.dominance}
        \PP(p_i \ge t \mid p_i \ge s) \ge (1-t)/(1-s), \quad q \le s \le t \le 1.
    \end{align}
    The conditional stochastic dominance property is met when the null p-values admit an increasing density on $[q,1]$, such as Beta$(a, b)$ with $a \ge 1$, $b \le 1$, and the p-values of one-sided hypothesis testing based on a distribution family with monotone likelihood ratio.
    See \Cref{sec:simulations} for concrete examples.
\end{remark}

\subsubsection{Dependent p-values based on exchangeable test statistics}\label{sec:RANC}

As discussed in \cite{benjamini2006adaptive, blanchard2008two}, the BH procedure with Storey's null proportion estimator does not control FDR in general.
Nevertheless, we prove the FDR control extends to a class of dependent p-values considered by \cite{bates21_testin_outlier_with_confor_p_values, mary22_semi_super_multip_testin, gao2023simultaneous}.
This type of p-values, also known as conformal p-values, are derived from exchangeable test statistics described below.
Their validity is robust to model misspecifications, small sample sizes, batch effects, and only requires an exchangeability condition that could be ensured by data generation or domain knowledge.

We describe this type of p-values following the notations in \cite{gao2023simultaneous}.
Suppose there are $\No + \NoNc$ hypotheses. 
The first $\No$ hypotheses $\calI := \{1,\dotsc,\No\}$ are under investigation, and the last $\NoNc$ hypotheses $\{\No+1,\dotsc,\No+m\}$ are true according to domain knowledge and referred to as negative controls.  We use $\calI_0 \subseteq \calI$ to denote the unknown null hypotheses under investigation.
Each hypothesis $\hypothesis{i}$ is
associated with a test statistic $\testStatistics{i}$, where a larger test statistic provides more evidence against that hypothesis. 
The p-value is defined as the
normalized rank of $\testStatistics{i}$ among $\testStatistics{i}$, $\testStatistics{\No +1}$, $\ldots$, $\testStatistics{\No + \NoNc}$,
\begin{align}\label{eq:pval}
  p_i := \frac{1 + \sum_{j \in \hypothesisIndex{\text{nc}}}
  \1_{\{\testStatistics{j} \ge \testStatistics{i}\}}}{1 + \NoNc}.
\end{align}
Because $p_i$, $i \in [n]$ are calculated using
the same set of negative controls, they are not independent even if the test statistics $\testStatistics{i}$, $i \in [\No +m]$ are themselves independent. 
Nevertheless, the BH procedure combined with $p_i$ with independent or partially exchangeable (specified below) test statistics has been proved to control FDR \parencite{bates21_testin_outlier_with_confor_p_values, mary22_semi_super_multip_testin, gao2023simultaneous}.
In this paper, we extend this FDR control result to the BH procedure with the AS null proportion estimator.

\begin{proposition}\label{prop:BH.adaptive.Storey.RANC}
    Let $\hat{\lambda}$ be a stopping time with regard to $\calG_t$.
    Assume the partial exchangeability:
    $\testStatistics{i}$, $i \in \calI_{0} \cup \calI_{\text{nc}}$ is exchangeable given $\testStatistics{j}$, $j \in \calI/\calI_{0}$.  Then \Cref{algo:BH.adaptive.Storey} applied to $(p_i)_{i \in \hypothesisIndex{}}$ defined in Eq.~\eqref{eq:pval} controls the FDR at the target level.
\end{proposition}

To prove \Cref{prop:BH.adaptive.Storey.RANC}, we first employ combinatorial equations to show that $M_t$ in \eqref{eq:martingale} based on the p-values in~\eqref{eq:pval} remains a super-martingale under the partial exchangeability condition.
We then provide a leave-one-out argument for the FDR control of the adaptive BH procedure.
This argument is distinct from existing proofs for this type of p-values relying on martingales \parencite{mary22_semi_super_multip_testin, gao2023simultaneous} or the PRDS property \parencite{bates21_testin_outlier_with_confor_p_values, gao2023simultaneous}

\subsection{A stopping time maximizing the number of rejections}\label{sec:stopping.time}

The general AS method accommodates arbitrary stopping time.
We provide a concrete configuration with the goal to maximize the number of rejections, or equivalently minimize $\hat{\pi}_0^{\hat{\lambda}}$.

Let $\lambda_j := q + j\delta$, $j \in [(1-q)/\delta]$ for some $\delta > 0$, we consider the stopping time
\begin{align}\label{eq:stopping.time}
    \hat{\lambda} = \inf\left\{\lambda_j, j \in [(1-q)/\delta]: \hat{\pi}_0^{\lambda_{j-1}} \le \hat{\pi}_0^{\lambda_j}\right\}.  
\end{align}
The stopping time compares the estimated null proportions of adjacent hyper-parameters $\lambda_{j-1}$, $\lambda_j$, and terminates once the estimated value stops to decrease, i.e., $\hat{\pi}_0^{\lambda_{j-1}} \le \hat{\pi}_0^{\lambda_j}$.
In \Cref{fig:protein.hist.intro} (b) derived from the motivating protein dataset, the estimator $\hat{\pi}_0^{\lambda}$ first decreases then increases in $\lambda$, and the proposed stopping time is likely to stop around the global minimizer of $\hat{\pi}_0^{\lambda}$.
Regarding the selection of $\delta$, a small $\delta$ is associated with $\lambda_j$ of higher resolution, but it may lead to early termination of the descending trend of $\ell$ due to the variability in $\hat{\pi}_0^{\lambda}$. 
In \Cref{prop:convergence.rate}, we analyze the asymptotic performance of $\hat{\lambda}$ considering a $\delta$ vanishing at the rate $O(\No^{-1/3})$.
In practice, we suggest using $\delta = 50/\sum_{i=1}^{\No} \1_{\{p_i \ge q\}}$ for $\No \sim 10^3$ as a rule of thumb to balance the low resolution and the premature termination.


\begin{figure}[tbp]
        \centering
        \begin{minipage}{0.33\textwidth}
                \centering
                \includegraphics[clip, trim = 0cm 0cm 0cm 0cm, height = 5cm]{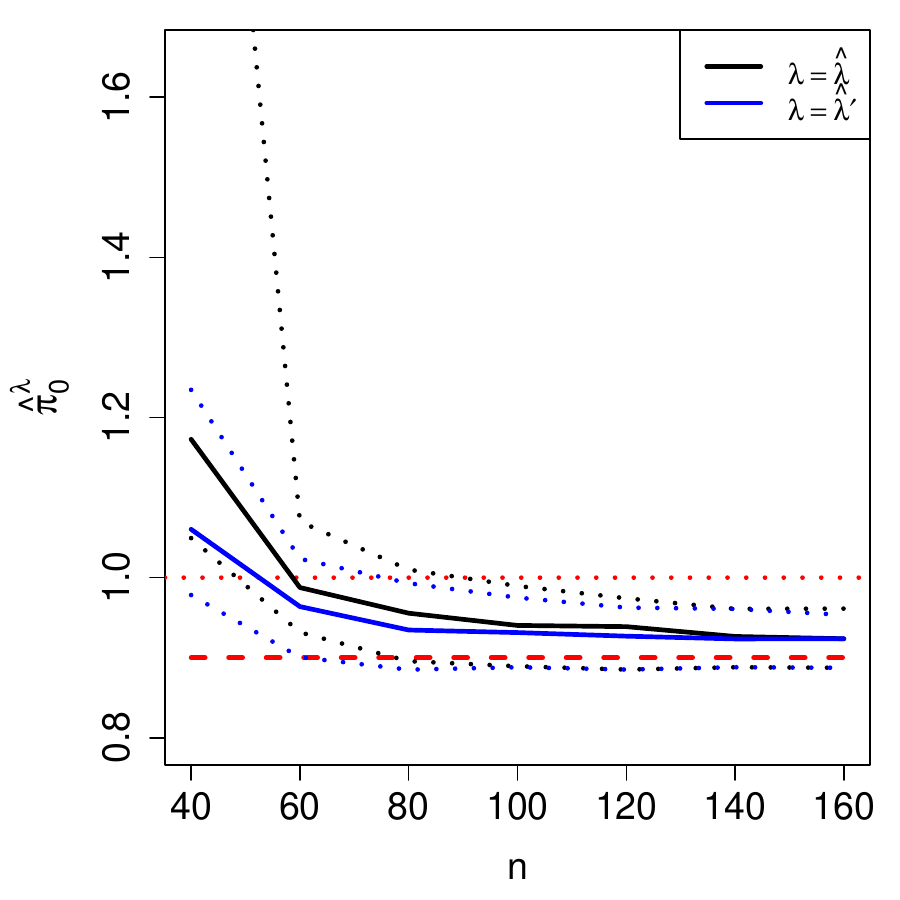}
        \subcaption{Few strong signals}
        \end{minipage}
        \begin{minipage}{0.32\textwidth}
                \centering
                \includegraphics[clip, trim = 0cm 0cm 0cm 0cm, height = 5cm]{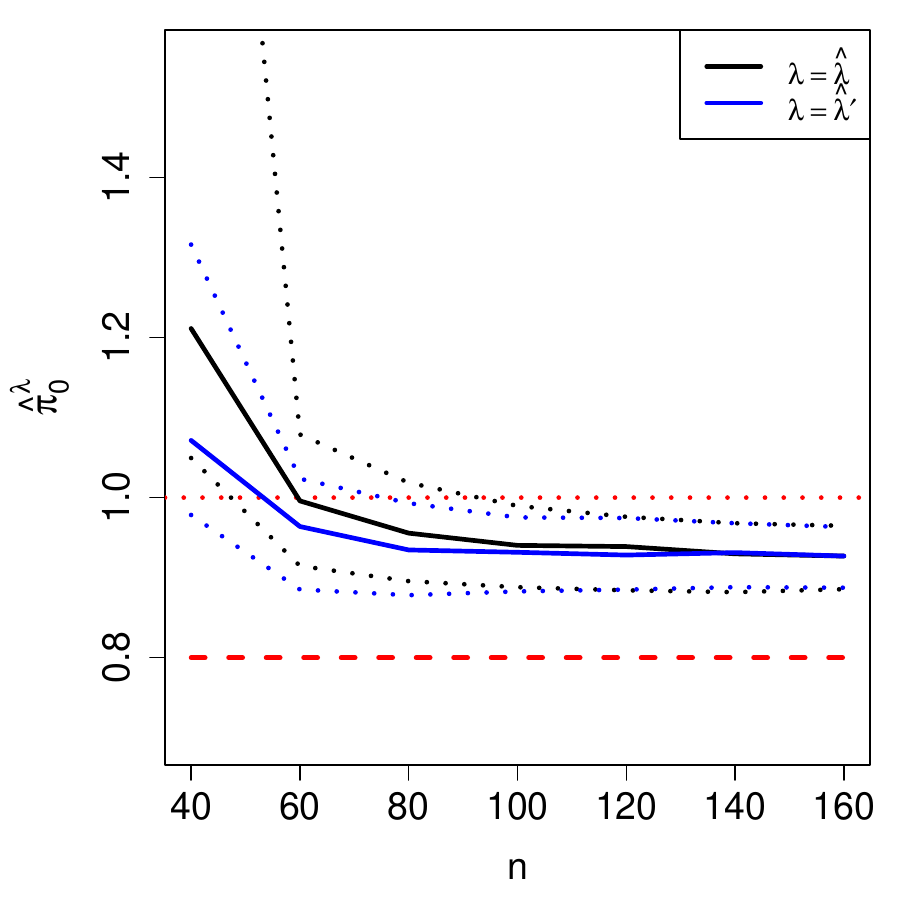}
        \subcaption{Many weak signals}
        \end{minipage}
        \begin{minipage}{0.32\textwidth}
                \centering
                \includegraphics[clip, trim = 0cm 0cm 0cm 0cm, height = 5cm]{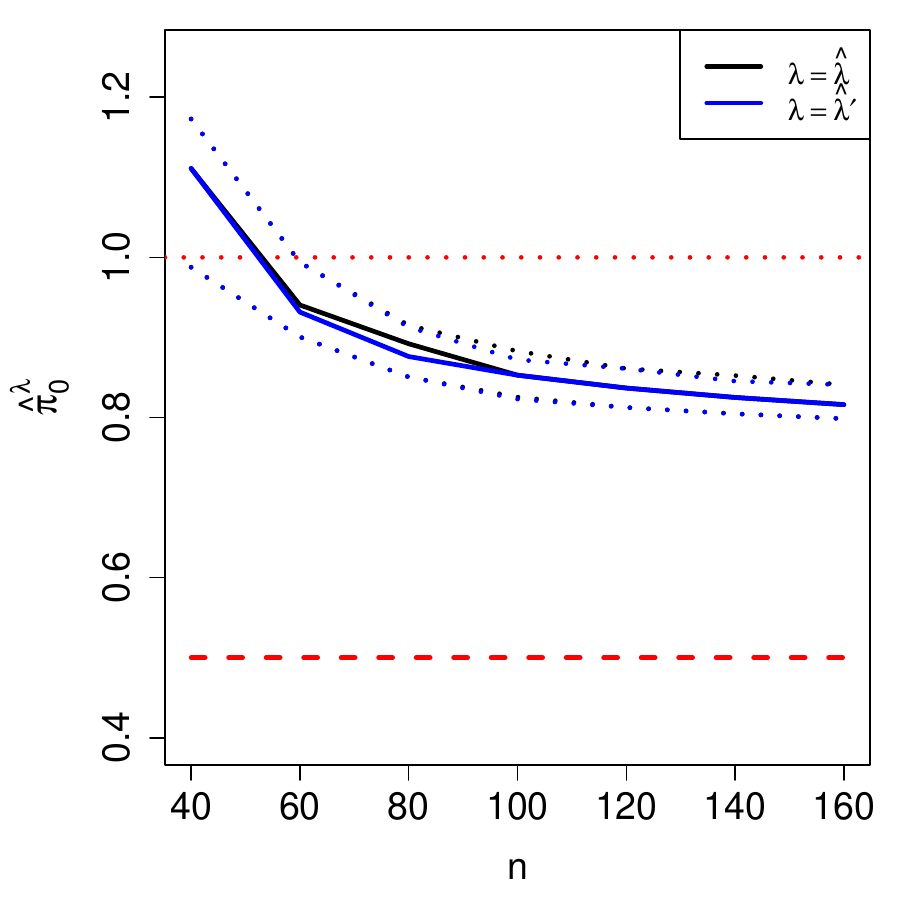}
        \subcaption{Conservative nulls}
        \end{minipage}
            \caption{Comparison of $\hat{\lambda}$ and its robust version $\hat{\lambda}'$. We consider the three scenarios in \Cref{sec:simulations} and vary the total sample size from $40$ to $160$. We choose $\delta = 10/\sum_{i=1}^{\No} \1_{\{p_i \ge q\}}$. Each scenario is repeated $4000$ times. We plot the median (solid), upper and lower quartiles (dotted) of Storey's estimator $\hat{\pi}_0^{\hat{\lambda}}$ (black), $\hat{\pi}_0^{\hat{\lambda}'}$ (blue). 
            We also superimpose the horizontal reference lines at the true $\pi_0$ (red, dashed) and $1$ (red, dotted). }
        \label{fig:lossCurvePi0}
\end{figure}

Next, we discuss a robust version of $\hat{\lambda}$ suitable for small sample size. Define a loss that takes the variation of $\hat{\pi}_0^{\lambda}$ into account,
\begin{align*}
    \ell({\lambda})
    := \hat{\pi}_0^{\lambda} + \sqrt{\hat{V}^{\lambda}}, \quad
    \hat{V}^{\lambda} := \frac{1}{n} \cdot \hat{\pi}_0^{\lambda} \left(\frac{1}{1 - \lambda} - \hat{\pi}_0^{\lambda} \right).
\end{align*}
Here $\hat{V}^{\lambda}$ estimates the variance of $\hat{\pi}_0^{\lambda}$ by plugging $\hat{\pi}_0^{\lambda}$ into the variance formula of Binomial distribution.
Close-to-one $\lambda$ is often associated with a large loss due to the standard deviation term.
We define a robust version of the stopping time~\eqref{eq:stopping.time} based on the loss $\ell({\lambda})$,
\begin{align}\label{eq:stopping.time.robust}
    \hat{\lambda}' = \inf\left\{\lambda_j, j \in [(1-q)/\delta]: \ell({\lambda}_{j-1}) \le \ell({\lambda}_j) \right\}.  
\end{align}
\Cref{fig:lossCurvePi0} displays that the estimated null proportion with the robust version of the stopping time could be more reliable when the sample size is small ($\No \le 100$).


Below we provide the asymptotic performance analysis of the proposed stopping times under a frequently-used mixture model \parencite{storey2004strong, efron2009empirical}. 
Suppose $(\hypothesis{i})_{i \in [\No]} \overset{\text{i.i.d.}}{\sim}
\text{Bernoulli}(1 - \probNull)$.
Given the hypothesis $\hypothesis{i}$, the
p-value is generated independently as
\begin{equation*}
  \label{eq:two-mixture}
  p_i \mid \hypothesis{i} = 0 \sim
  \cdfTestStatisticsNull \quad \text{and} \quad p_i
  \mid \hypothesis{i} = 1 \sim \cdfTestStatisticsNonNull.
\end{equation*}
Here $F_0$, $F_1$ denote the CDF of the unknown null\footnote{Here we allow an unknown null distribution that conditionally stochastically dominates $U(0,1)$ as described in \Cref{rmk:stochastic.dominance}.} and non-null distributions. 
The marginal CDF of $p_i$ is the mixture
$\cdfTestStatistics{} = \probNull \cdfTestStatisticsNull +
(1-\probNull) \cdfTestStatisticsNonNull$.
We denote the density functions
of $\cdfTestStatisticsNull$ and
$\cdfTestStatisticsNonNull$ by $f_0$, $f_1$, and the marginal density by $f = \pi_0 f_0 + (1-\pi_0) f_1$.
Let $\pi^{\lambda}_{0,\infty} = (1 - F(\lambda))/(1-\lambda)$, $0 \le \lambda < 1$, and $\pi^{1}_{0,\infty} = \overline{\lim}_{t \to 1} \pi^{t}_{0,\infty}$, which can be considered as the null proportion estimator with the hyper-parameter $\lambda$ obtained from an infinite number of hypotheses drawn from the mixture model.
Define $\pi_0^* := \inf_{q \le \lambda \le 1} \pi^{\lambda}_{0,\infty} \ge \pi_0$, and the optimal hyper-parameter $\lambda^* := \inf\{q \le \lambda \le 1: \pi^{\lambda}_{0,\infty} = \pi^*_0\}$. 
It is possible that the optimal hyper-parameter $\lambda^* < 1$, typically in the presence of conservative nulls as in the motivating protein dataset.

Under the mixture model, Storey's estimator with a fixed $\lambda < 1$ is shown to converge at the rate $\No^{-1/2}$ to a Gaussian distribution with the mean $\pi_{0,\infty}^{\lambda} \ge \pi_0^*$ \parencite{genovese_stochastic_2004}.
\cite{neuvial2013asymptotic} shows Storey's estimator with a hyper-parameter growing to $1$ converges at a rate slower than $\No^{-1/2}$.
In contrast, our analysis below deals with data-driven hyper-parameters.

Let $\hat{\lambda}_{\No}$, $\hat{\lambda}_{\No}'$ be the stopping time defined in \eqref{eq:stopping.time}, \eqref{eq:stopping.time.robust} obtained from a sample of size $\No$. The following proposition characterizes the convergence of $\hat{\lambda}_{\No}$, $\hat{\lambda}_{\No}'$ to $\lambda^*$.

\begin{proposition}\label{prop:convergence.rate}
    Assume
    $f$, $f'$, $f''$ exist, are continuous on $[0,1]$, and $f'' > 0$ on $[0,1)$.
    \begin{itemize}
        \item [(a)] If $\lambda^* = 1$, then for $\delta = \No^{-1/3} \log(\No)$, \begin{align*}
        \PP\left(\left|\hat{\lambda}_{\No} - \lambda^*\right| \ge \varepsilon \right),  
        ~\PP\left(\hat{\pi}_0^{\hat{\lambda}_{\No}} - {\pi}_0^* \ge \varepsilon\right) \to 0.
    \end{align*}
        \item [(b)]  If $\lambda^* < 1$, then for any $M > 0$, there exists $N > 0$ such that setting $\delta = M \No^{-1/3}$, with probability at least $1 - C_1 M^{-3/2}$, 
    \begin{align*}
        \left| \hat{\lambda}_{\No} - \lambda^* \right| \le C_2 M {\No}^{-1/3}, \quad 
        \hat{\pi}_0^{\hat{\lambda}_{\No}} \le \pi_0^* + C_2 M {\No}^{-1/3}, \quad \forall \No \ge N,
    \end{align*}
    where $C_1$, $C_2 > 0$ are independent of $\No$ and $M$. 
    \end{itemize}
    Both (a) and (b) apply to $\hat{\lambda}_{\No}'$.
\end{proposition}

The proof is provided in the appendix.
Examples satisfying the conditions in \Cref{prop:convergence.rate} include p-values derived from a family with monotone likelihood ratio.
If the condition is violated in the sense that $f'' = 0$ on an interval contained in $(q, \lambda^*)$, then $\hat{\lambda}_{\No}$ is likely to terminate within that interval before reaching $\lambda^*$ and is thus not consistent.

\section{Simulations}\label{sec:simulations}

We compare the non-adaptive BH (std), the oracle BH (orc) with true $\probNull$, and the adaptive BH with three types of null proportion estimators: BY for \cite{benjamini2006adaptive}, S for Storey's estimator, and AS for the proposed extension.
For Storey's estimator, we consider three variants: S.S (S for small)  with $\lambda = 0.2$, S.M (M for medium) with $\lambda = 0.5$, S.L (L for large) with $\lambda = 0.8$.
In the AS method, we adopt a modified version of the stopping time~\eqref{eq:stopping.time.robust} by truncating it at $0.8$,
\begin{align*}
    \hat{\lambda} = \inf\left\{\lambda_j, j \in [(1-q)/\delta]: \lambda_j \le 0.8, ~{\ell}({\lambda}_{j-1}) \le {\ell}({\lambda}_{j})\right\}, 
\end{align*}
and the suggested rule of thumb $\delta = 50/\sum_{i=1}^{\No} \1_{\{p_i \ge q\}}$.
The FDR level is set at $q = 0.2$.

We consider four data generating mechanisms.
\begin{enumerate}
    \item [(a)] Few strong signals. We set $\No = 500$ with $\NoNull = 450$. The non-null p-values are generated independently as $1 - \Phi(Z + 2)$, where $Z \sim \calN(0,1)$ and $\Phi$ denotes the CDF of standard normal; 
    \item [(b)] Distributed weak signals. We consider $\No = 500$, $\NoNull = 100$, and non-null p-values $1 - \Phi(Z + 1.5 \mu_i)$ for $\mu_i = i/400$, $i \in [400]$; 
    \item [(c)] Screened using external data. We adopt $\No = 10^4$ with $\NoNull = 9 \times 10^3$. We generate two independent batches of p-values with non-null p-values following $1 - \Phi(Z + 1.5)$ and $1 - \Phi(Z + 1)$, respectively. We use the first batch to select the hypotheses with p-values no larger than $0.05$. We then apply BH to the p-values of the selected hypotheses in the second batch, whose null proportion is expected to be low due to the screening. This setting is motivated by the cross-screening procedure in \cite{bogomolov2018assessing};
    \item [(d)] Conservative nulls. We consider $\No = 500$, $\NoNull = 250$, and null p-values following Beta$(3, 1)$, non-null p-values $1 - \Phi(Z + 2)$. We remark that \Cref{prop:BH.adaptive.Storey} holds for this null distribution. 
\end{enumerate}

\begin{figure}[tbp]
        \centering
        \begin{minipage}{0.245\textwidth}
                \centering
                \includegraphics[clip, trim = 0cm 0cm 0cm 0cm, height = 16cm]{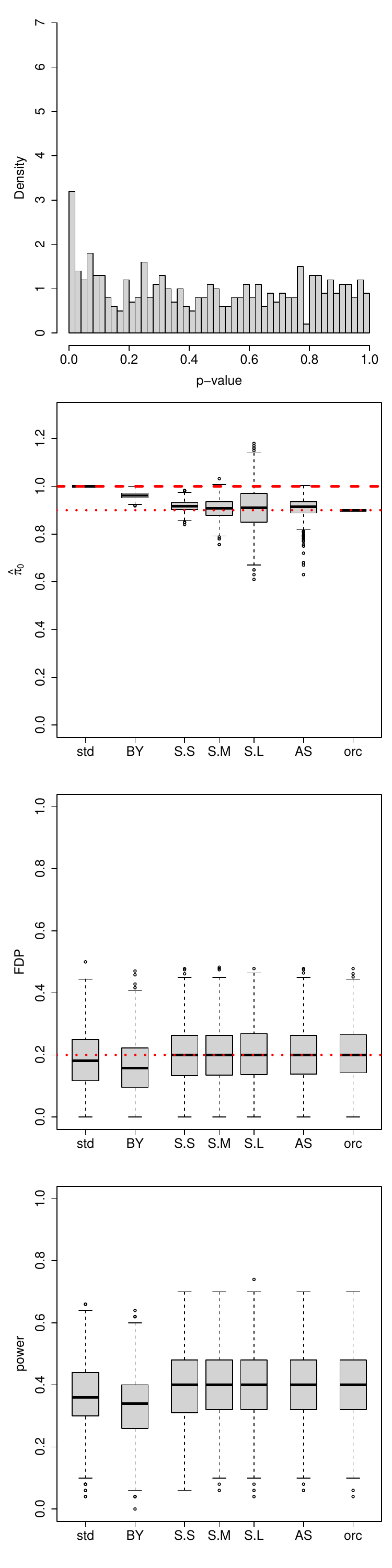}
        \subcaption{Few strong signals}
        \end{minipage}
        \begin{minipage}{0.21\textwidth}
                \centering
                \includegraphics[clip, trim = 1.5cm 0cm 0cm 0cm, height = 16cm]{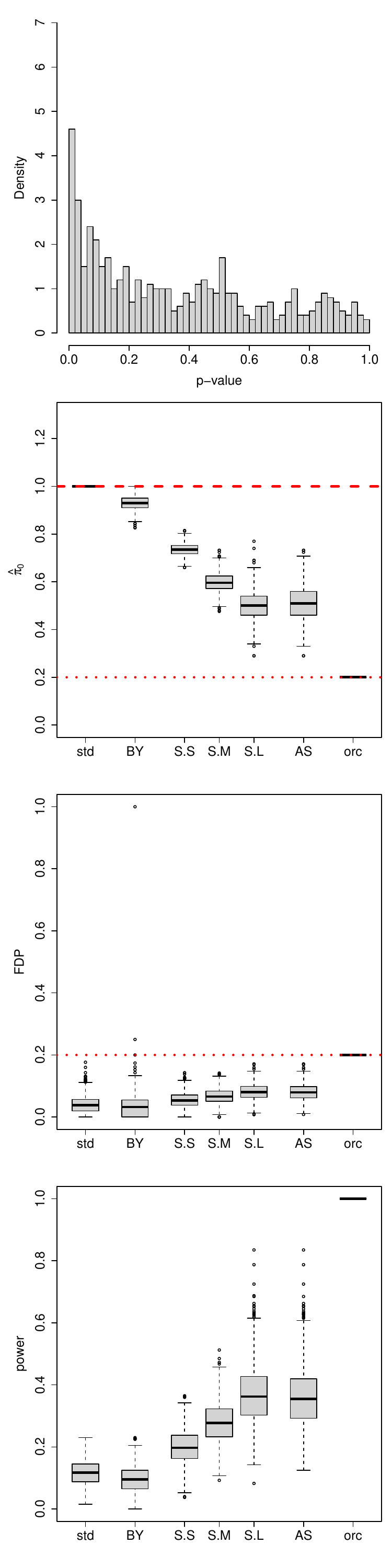}
        \subcaption{Many weak signals}
        \end{minipage}
        \begin{minipage}{0.2\textwidth}
                \centering
                \includegraphics[clip, trim = 1.5cm 0cm 0cm 0cm, height = 16cm]{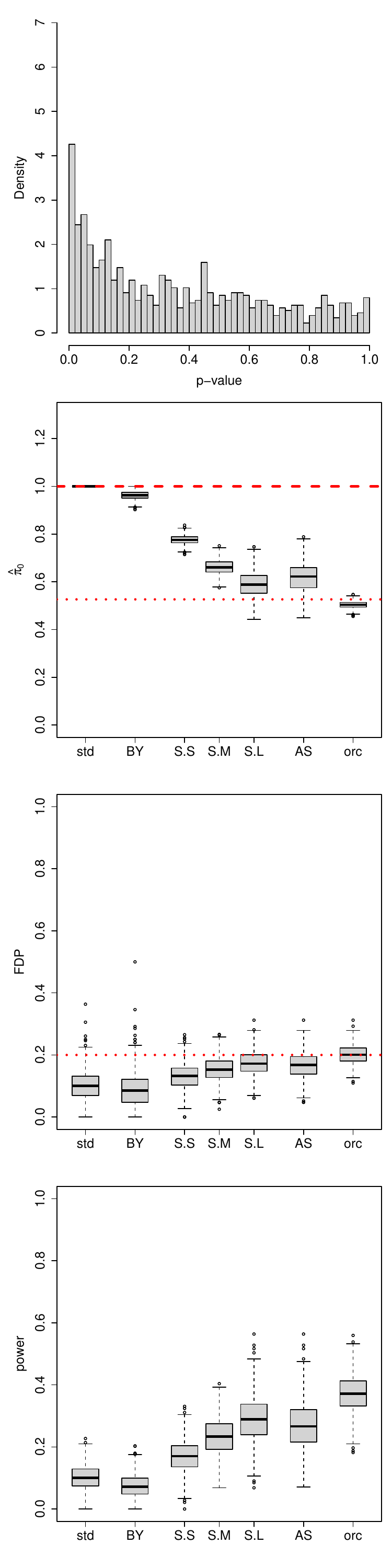}
        \subcaption{Screened}
        \end{minipage}
        \begin{minipage}{0.24\textwidth}
                \centering
                \includegraphics[clip, trim = 1.5cm 0cm 0cm 0cm, height = 16cm]{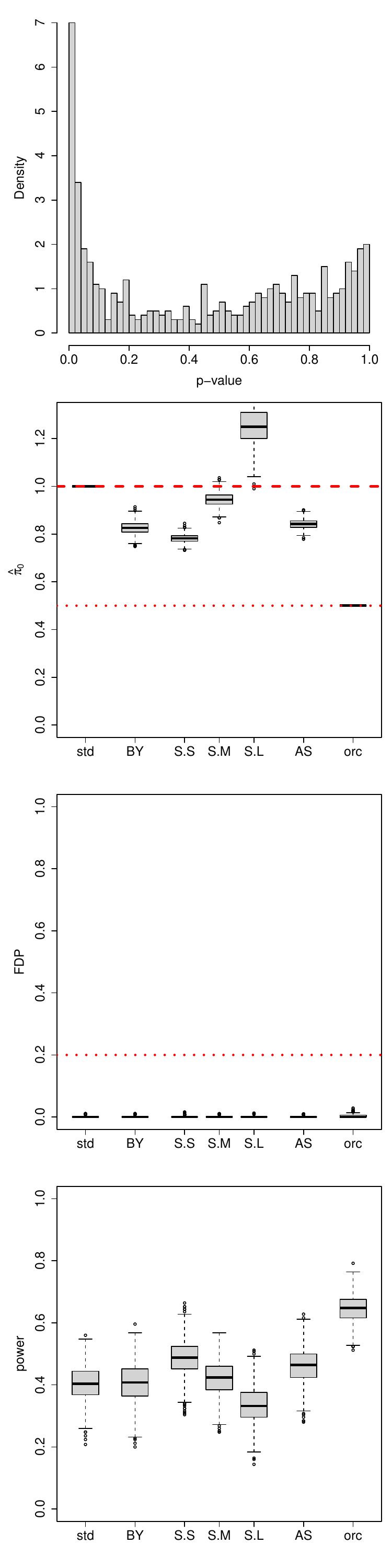}
        \subcaption{Conservative nulls}
        \end{minipage}

            \caption{Comparison of variants of the BH procedure regarding FDR and power.
            The first panel provides the histogram of the p-values from a randomly selected trial of each scenario.
            The second panel demonstrates the estimated null proportion, which is one for the standard BH procedure and $\pi_0$ for the oracle method. 
            The third and the fourth panels display the FDP and the proportion of correct rejections among non-nulls, respectively. 
            We experiment with the four scenarios described in \Cref{sec:simulations}, each repeated $1000$ times.}
        \label{fig:simulation}
\end{figure}

As depicted in \Cref{fig:simulation}, all methods control FDR throughout the four settings.
Regarding power, our proposed method consistently demonstrates high levels of power, while the performance of Storey's method with different pre-fixed $\lambda$ varies across settings, and BY, std generally make less rejections.
More explicitly, in (a), the proportion of null hypotheses is close to one, and all methods perform approximately the same;
in (b) and (c) with a significant proportion of weak non-nulls, the proposed method exhibits higher power than BY, S.S, and S.M with a small or medium $\lambda$;
in (d), S.M, S.L with a medium, large $\lambda$ only make a few rejections due to the conservative null distribution (Beta$(3,1)$), while the proposed method manages to maintain high power.

\section{Real data analysis}\label{sec:real.data}


We revisit the motivating dataset discussed in \Cref{sec:introduction}. The histogram of p-values (\Cref{fig:protein.hist.intro} (a)) displays two characteristics: there are a significant proportion of weak signals, which contribute to the gradual decrease in frequency for p-values in $[0.1, 0.7]$. This pattern is reflected in the simulation setting (b) where a large $\lambda$ is preferred; the null p-values are conservative and result in an increase in frequency at the right end of the p-values' range ($p \in [0.8, 1]$). This trend is reminiscent of the simulation setting (d) where a small $\lambda$ is preferred. 
The presence of two patterns, each favoring a different $\lambda$, adds the complexity to the choice of an appropriate hyper-parameter.

We apply the six methods in \Cref{sec:simulations} at the FDR level $q = 0.2$ to the protein dataset with all the hyper-parameters kept the same.
In \Cref{tab:protein}, we observe that the proposed method with the proposed stopping time~\eqref{eq:stopping.time} adaptively chooses ${\lambda}_{\text{AS}} = 0.68$, close to the global minimizer of $\hat{\pi}_0^{\lambda}$.
In contrast, variants of Storey's method with pre-defined hyper-parameters $\lambda \in \{0.2, 0.5, 0.8\}$ miss the optimal hyper-parameter.
This leads to larger estimates of the null proportion and fewer rejections. 
Finally, the BY procedure makes less rejections than the standard BH since the signals are not sufficiently strong to be picked up by its null proportion estimator.
Furthermore, we perform gene ontology analysis of cellular components \parencite{szklarczyk2023string} and among the $14$ proteins uniquely identified by our proposed method, $8$ proteins---Scn2a1, Itgb1, Lrrc4b, Slc6a9, Kcna2, Igsf21, Hapln4, Pde2a---are expected to be enriched.

\begin{table}[h]
\centering
\begin{tabular}{c|c|c|ccc|c}
\toprule
 & standard & BY & \multicolumn{3}{c|}{Storey} & AS \\
 &  &  &  $\lambda = 0.2$     &   $\lambda = 0.5$    &  $\lambda = 0.8$    &  \\ \midrule
$\hat{\pi}_0$ & $1$ & $1.17$  & $0.85$  & $0.69$      &   $0.72$   & $0.62$ \\
$R$ & $16$ & $12$ & $16$      &   $17$    &  $17$    & $31$ \\ \bottomrule
\end{tabular}
\caption{Estimated null proportion $\hat{\pi}_0$ and the number of rejections $R$ of variants of the BH procedure obtained from the motivating protein dataset.}
\label{tab:protein}
\end{table}


\section{Discussions}\label{sec:discussion}

In \cite{benjamini2006adaptive}, it is noted that Storey's method does not control FDR for dependent p-values with $\lambda = 0.5$, and  \cite{blanchard2009adaptive} show that a smaller $\lambda$ controls the FDR for equi-correlation model.
It would be valuable to investigate whether our adaptive extension is capable of handling dependent p-values beyond the conformal p-values in \Cref{sec:RANC}.

Choosing the stopping time to maximize the power has connections to optimal stopping time theory. 
However, in traditional optimal stopping problems, the context, that is $F$, $F_0$, $F_1$, and $\probNull$ in the two-component mixture model, is typically known in advance, while such information is not available and we learn it as we progress.
It is of interest to investigate whether an optimal stopping time exists in our setting and, if so, how to construct it.


\section*{Acknowledgement}
We would like to express our sincere gratitude to Qingyuan Zhao for the valuable and insightful suggestions regarding the statistical methodology. We genuinely thank Zora Chan for generously providing the protein dataset and assisting with the analysis therein.

\printbibliography

\appendix

\section{Proofs}\label{sec:proofs}

\subsection{FDR control for independent p-values}\label{sec:proof.independent}

\begin{proof}[Proof of \Cref{prop:BH.adaptive.Storey}]
    Recall the filtration
    \begin{align*}
        \calF_{t} = \sigma\left( \sum_{H_i = 0} \1_{\{p_i \ge s\}}, \sum_{i \in [n]} \1_{\{p_i \ge s\}}: q \le s \le t\right), \quad q \le t \le 1,
    \end{align*}
    and the random process
    \begin{align*}
        M_t
        = \frac{1 - t}{1 + \sum_{H_i = 0} \1_{\{p_i \ge t\}}}, \quad q \le t \le 1.
    \end{align*}
    which is non-negative and adapts to $\calF_t$.
    Since $p_i \indep \bm p_{-i}$ for $H_i = 0$, then for $q \le s \le t \le 1$,
    \begin{align*}
        \EE\left[  M_t \mid \calF_s \right]
        &= \EE\left[  \frac{1 - t}{1 + \sum_{H_i = 0} \1_{\{p_i \ge t\}}} \mid \sum_{H_i = 0} \1_{\{p_i \ge s\}} \right].
    \end{align*}
    Under the condition~\eqref{eq:conditional.stochastic.dominance}, 
    \begin{align*}
        \sum_{H_i = 0} \1_{\{p_i \ge t\}} \mid \sum_{H_i = 0} \1_{\{p_i \ge s\}} = K ~\succeq~ \text{Binomial}\left(K, \frac{1-t}{1-s}\right),
    \end{align*} 
    here $\succeq$ denotes stochastic dominance.
    Together with the fact that $\EE[1/(1+X)] \le ((1+K)p)^{-1}$ for $X \sim \text{Binomial}(K, p)$, we obtain
    \begin{align*}
        \EE\left[  M_t \mid \calF_s \right]
        &\le \frac{1 - t}{\left(1 +  \sum_{H_i = 0} \1_{\{p_i \ge s\}}\right) \cdot ({1-t})/({1-s})}
        = M_s,
    \end{align*}
    which implies $M_t$ is a non-negative super-martingale w.r.t. $\calF_t$.
    By the optional stopping time theorem, for any stopping time $\hat{\lambda}$ taking values in $[q,1]$,
    \begin{align*}
         \EE\left[M_{\hat{\lambda}} \right]
         \le \EE\left[M_q \right]
         \le \frac{1-q}{(1 + \NoNull) \cdot (1 - q)}
          = \frac{1}{1 + \NoNull},
    \end{align*}
    where in the second inequality we again use the condition~\eqref{eq:conditional.stochastic.dominance} and the property of Binomial distribution.
    Finally, 
    \begin{align}\label{proof:eq:expectation}
        \EE\left[ \frac{1}{\hat{\pi}_0^{\hat{\lambda}}}\right] 
        = \EE\left[\frac{n\left(1 - \hat{\lambda}\right)}{1 + \sum_{i=1}^n \1_{\{p_i \ge \hat{\lambda}\}}}\right]
        \le \EE\left[\frac{n\left(1 - \hat{\lambda}\right)}{1 + \sum_{H_i=0} \1_{\{p_i \ge \hat{\lambda}\}}}\right]
        = \No \cdot \EE\left[  M_{\hat{\lambda}} \right]
        \le \frac{n}{1+\NoNull}.   
    \end{align}

    Now we prove the BH procedure combined with $\hat{\pi}_0^{\hat{\lambda}}$ controls FDR by the leave-one-out argument.
    For $H_i = 0$, we use $V_i = 1$ to denote $H_i$ being rejected.    
    Define
    \begin{align*}
    \hat{\pi}_0^{\hat{\lambda},-i} := \inf\left\{\hat{\pi}_0^{\hat{\lambda}}(\bm p'): \bm p'_{-i} = \bm p_{-i}, p'_i < q\right\}.
    \end{align*}
    Notice that $\hat{\lambda}$ adapts to $\calF_t$ and does not depend on the p-values smaller than $q$, then $\No \hat{\pi}_0^{\hat{\lambda},-i} = (\No-1)\hat{\pi}_0^{\hat{\lambda}}(\bm p_{-i})$. 
    Since $p_i \indep \bm p_{-i}$, Eq.~\eqref{proof:eq:expectation} applied to $\hat{\pi}_0^{\hat{\lambda}}(\bm p_{-i})$ implies 
    \begin{align}\label{proof:eq:expectation.2}
        \EE\left[ \frac{1}{\hat{\pi}_0^{\hat{\lambda}, -i}}\right] 
        = \frac{\No}{\No - 1} \cdot \EE\left[ \frac{1}{\hat{\pi}_0^{\hat{\lambda}}(\bm p_{-i})}\right] 
        \le \frac{\No}{\No - 1} \cdot \frac{n - 1}{1+(\NoNull-1)}
        = \frac{n}{\NoNull}.   
    \end{align}
    Define
    \begin{align*}
        R^{-i}(\bm p_{-i}) = \sup\left\{|R(\bm p')|: \bm p'_{-i} = \bm p_{-i}, V_i = 1\right\},
    \end{align*}
    Then the adaptive BH procedure satisfies $R = R^{-i}$ given $V_i = 1$, $\bm p_{-i}$.
    Since $V_i = 1$ implies $q_i < q$, we have $\hat{\pi}_0^{\hat{\lambda}} = \hat{\pi}_0^{\hat{\lambda},-i}$ given $V_i = 1$, $\bm p_{-i}$.
    Therefore, 
    \begin{align*}
        \EE\left[\frac{V_i}{R \vee 1}\right]
        &= \EE\left[\frac{V_i \1_{\left\{p_i \le {Rq}/{\left(n \hat{\pi}_0^{\hat{\lambda}}\right)}, p_i < q\right\}}}{R \vee 1}\right]\\
        &= \EE\left[\frac{V_i \1_{\left\{p_i \le {R^{-i}q}/{\left(n \probNull^{\hat{\lambda}, -i}\right)}, p_i < q\right\}}}{R^{-i} \vee 1}\right]\quad (R = R^{-i}, \hat{\pi}_0^{\hat{\lambda}} = \hat{\pi}_0^{\hat{\lambda},-i} ~\text{given}~ V_i = 1, \bm p_{-i})\\
        &\le \frac{q}{n} \cdot \EE\left[\frac{1}{ \probNull^{\hat{\lambda}, -i}}\right]  \quad (p_i \indep \bm p_{-i}) \\
        &\le \frac{q}{n} \cdot \frac{n}{\NoNull}
        = \frac{q}{\NoNull}. \quad (\text{Eq.}~\eqref{proof:eq:expectation.2})
    \end{align*}
    We sum over $H_i = 0$ and finish the proof. 
\end{proof}

\subsection{FDR control for conformal p-values}\label{sec:proof.RANC}

\begin{proof}[Proof of \Cref{prop:BH.adaptive.Storey.RANC}]
    We extend the definitions of $\calF_{t}$ and $M_t$ in the proof of \Cref{prop:BH.adaptive.Storey} to $1/(\NoNc+1) \le t \le 1$.
    Let $t_{\NoNc - k} := 1 - k/(\NoNc + 1)$.
    Under the partial exchangeability condition, for $q \le t_{\NoNc - k} < t_{\NoNc - k + 1} \le 1$, 
    \begin{align*}
        \EE\left[M_{t_{\NoNc - k + 1}} \mid \calF_{t_{\NoNc - k}} \right]
        &= \EE\left[  \frac{k/(\NoNc+1)}{1 + \sum_{H_i = 0} \1_{\{p_i \ge t_{\NoNc - k + 1}\}}} \mid \sum_{H_i = 0} \1_{\{p_i \ge t_{\NoNc - k}\}} \right].
    \end{align*}
    Conditional on $\sum_{H_i = 0} \1_{\{p_i \ge t_{\NoNc - k}\}} = K$, the $(\NoNc - k + 1)$-th smallest negative control test statistic ranks $r \in [K+1]$ among the largest $K$ test statistics under investigation and itself with probability
    \begin{align}\label{proof:eq:rank.probability}
         {\binom{K + k - r}{k - 1}} \cdot {\binom{K + k}{k}}^{-1}
    \end{align}
    under the partial exchangeability.
    Therefore, 
    \begin{align*}
         \EE\left[M_{t_{\NoNc - k + 1}} \mid \sum_{H_i = 0} \1_{\{p_i \ge t_{\NoNc - k}\}} = K \right]
         &= \sum_{r=1}^{K+1} \frac{(k-1)/(\NoNc+1)}{1 + K - (r-1)} \cdot \frac{\binom{K + k - r}{k - 1}}{\binom{K + k}{k}} \\
         &= \frac{1}{\NoNc+1} \cdot \frac{1}{\binom{K + k}{k}} \sum_{r=1}^{K+1} {\binom{K + k - r}{k - 2}} \\
         &= \frac{1}{\NoNc+1} \cdot \frac{1}{\binom{K + k}{k}} \left(\binom{K + k}{k - 1} - 1\right) \quad \left(\sum_{a = b}^{B} \binom{a}{b} = \binom{B+1}{b+1} \right) \\
         &< \frac{1}{\NoNc+1} \cdot \frac{k}{K+1}
         = M_{t_{\NoNc - k}}.
    \end{align*}
    Together with the fact that $M_s \le M_{t_{\NoNc - k}}$ for $t_{\NoNc - k} < s \le t_{\NoNc - k + 1}$, we conclude $M_t$ is a non-negative super-martingale w.r.t. $\calF_t$.
    For $q \ge 1/(\NoNc+1)$,
    \begin{align*}
        \EE\left[M_{q}\right] 
        \le \EE\left[M_{1/(\NoNc+1)}\right] 
        < \frac{1}{\NoNc+1} \cdot \frac{\NoNc}{\NoNull+1}
         < \frac{1}{\NoNull+1}.
    \end{align*} 
    From here, we can follow the proof of \Cref{prop:BH.adaptive.Storey} and employ the optional stopping theorem to show $\EE[ {(\hat{\pi}_0^{\hat{\lambda}})^{-1}}] \le {n}/{(1+\NoNull)}$.

    Now we extend the leave-one-out argument to $p_i$.  
    For each $i \in \hypothesisIndex{}$, we define 
    \begin{align*}
        p_{j}^{-i}(k;\testStatistics{\hypothesisIndex{}-\{i\}},\{\testStatistics{\hypothesisIndex{\text{nc}} \cup \{i\}}\}) &:= \pval{j}\left(\testStatistics{i}~\text{ranks the $(k+1)$-th in}~\{\testStatistics{\hypothesisIndex{\text{nc}} \cup \{i\}}\}; \testStatistics{\hypothesisIndex{}-\{i\}}, \{\testStatistics{\hypothesisIndex{\text{nc}} \cup \{i\}}\}\right), ~j \in \hypothesisIndex{}, \\
        R^{-i}(k;\testStatistics{\hypothesisIndex{}-\{i\}},\{\testStatistics{\hypothesisIndex{\text{nc}} \cup \{i\}}\}) &:= R\left(\testStatistics{i}~\text{ranks the $(k+1)$-th in}~\{\testStatistics{\hypothesisIndex{\text{nc}} \cup \{i\}}\}; \testStatistics{\hypothesisIndex{}-\{i\}}, \{\testStatistics{\hypothesisIndex{\text{nc}} \cup \{i\}}\}\right),\\
        \hat{\pi}_0^{\hat{\lambda},-i}(k;\testStatistics{\hypothesisIndex{}-\{i\}},\{\testStatistics{\hypothesisIndex{\text{nc}} \cup \{i\}}\}) &:= \hat{\pi}_0^{\hat{\lambda}}\left(\testStatistics{i}~\text{ranks the $(k+1)$-th in}~\{\testStatistics{\hypothesisIndex{\text{nc}} \cup \{i\}}\}; \testStatistics{\hypothesisIndex{}-\{i\}}, \{\testStatistics{\hypothesisIndex{\text{nc}} \cup \{i\}}\}\right),\\
        \kappa_i(\testStatistics{\hypothesisIndex{}-\{i\}},\{\testStatistics{\hypothesisIndex{\text{nc}} \cup \{i\}}\}) &:= \max \left\{-1 \le k < q(\NoNc+1) - 1: \frac{1+k}{1 + \NoNc} \le \frac{q R^{-i}(0;\testStatistics{\hypothesisIndex{}-\{i\}},\{\testStatistics{\hypothesisIndex{\text{nc}} \cup \{i\}}\})}{\No \hat{\pi}_0^{\hat{\lambda},-i}(0;\testStatistics{\hypothesisIndex{}-\{i\}},\{\testStatistics{\hypothesisIndex{\text{nc}} \cup \{i\}}\})}\right\}.
    \end{align*}
    We drop $\testStatistics{\hypothesisIndex{}-\{i\}}$, $\{\testStatistics{\hypothesisIndex{\text{nc}} \cup \{i\}}\}$ and use $p_{j}^{-i}(k)$, $R^{-i}(k)$, $\hat{\pi}_0^{\hat{\lambda},-i}(k)$, $\kappa_i$ for notation simplicity.
    Notice that $\hat{\lambda}$ adapts to $\calF_t$ and does not depend on the p-values smaller than $q$, then $\hat{\pi}_0^{\hat{\lambda}}= \hat{\pi}_0^{\hat{\lambda},-i}(0) \mid V_i = 1$, $\testStatistics{\hypothesisIndex{}-\{i\}}, \{\testStatistics{\hypothesisIndex{\text{nc}} \cup \{i\}}\}$.
    Under the partial exchangeability condition, given $p_i = 1/(\NoNc + 1) < q$, Eq.~\eqref{proof:eq:rank.probability} holds and $\EE[(\hat{\pi}_0^{\hat{\lambda}, -i}(0))^{-1}] \le {n}/{\NoNull}$.
    We inherit the notation $V_i$. For $H_i = 0$,
    \begin{align*}
        \EE\left[\frac{V_i}{R \vee 1} \right]
        &= \EE\left[\EE\left[\frac{V_i}{R \vee 1} \mid \testStatistics{\hypothesisIndex{}-\{i\}}, \{\testStatistics{\hypothesisIndex{\text{nc}} \cup \{i\}}\}\right]\right] \quad (\text{tower property})\\
        &= \EE\left[\EE\left[\frac{V_i \1_{\{\pval{i} < q, \pval{i} \le q R/(\No \hat{\pi}_0^{\hat{\lambda}, -i}(0))  \}}}{R \vee 1} \mid \testStatistics{\hypothesisIndex{}-\{i\}}, \{\testStatistics{\hypothesisIndex{\text{nc}} \cup \{i\}}\}\right]\right] \quad (\hat{\pi}_0^{\hat{\lambda}, -i}(0) = \hat{\pi}_0^{\hat{\lambda}} ~\text{if}~ p_i < q) \\
        &= \EE\left[\EE\left[\frac{V_i \1_{\{\pval{i} < q, \pval{i} \le q R^{-i}(0)/(\No \hat{\pi}_0^{\hat{\lambda}, -i}(0))  \}}}{R^{-i}(0) \vee 1} \mid \testStatistics{\hypothesisIndex{}-\{i\}}, \{\testStatistics{\hypothesisIndex{\text{nc}} \cup \{i\}}\}\right]\right] \quad (\text{\Cref{lemm:RANC.stable}})\\
        &\le \EE\left[\EE\left[\frac{\1_{\{\pval{i} < q, \pval{i} \le q R^{-i}(0)/(\No \hat{\pi}_0^{\hat{\lambda}, -i}(0))  \}}}{R^{-i}(0) \vee 1} \mid \testStatistics{\hypothesisIndex{}-\{i\}}, \{\testStatistics{\hypothesisIndex{\text{nc}} \cup \{i\}}\}\right]\right] \quad (V_i \le 1)
        \\
        &= \EE\left[\frac{\EE\left[ \1_{\{k \le \kappa_i\}} \mid \testStatistics{\hypothesisIndex{}-\{i\}}, \{\testStatistics{\hypothesisIndex{\text{nc}} \cup \{i\}}\}\right]}{R^{-i}(0) \vee 1 }  \right] \quad (R^{-i}(0) \sim \sigma(\testStatistics{\hypothesisIndex{}-\{i\}}, \{\testStatistics{\hypothesisIndex{\text{nc}} \cup \{i\}}\})).
    \end{align*}
    Since $\testStatistics{\hypothesisIndex{\text{nc}} \cup \{i\}}$ is exchangeable conditional on $\testStatistics{\hypothesisIndex{} - \{i\}}$, 
    then 
    \begin{align*}
        \EE\left[ \1_{\{k \le \kappa_i\}} \mid \testStatistics{\hypothesisIndex{}-\{i\}}, \{\testStatistics{\hypothesisIndex{\text{nc}} \cup \{i\}}\}\right] = \frac{1+\kappa_i}{1+\NoNc}
        \le \frac{q R^{-i}(0)}{\No \hat{\pi}_0^{\hat{\lambda}, -i}(0)}.
    \end{align*}
    Finally we use the upper bound of $\EE[(\hat{\pi}_0^{\hat{\lambda}, -i}(0))^{-1}]$ and sum over $H_i = 0$.
\end{proof}

\begin{lemma}\label{lemm:RANC.stable}
    Given $\testStatistics{\hypothesisIndex{}-\{i\}}$, $\{\testStatistics{\hypothesisIndex{\text{nc}} \cup \{i\}}\}$, and suppose $\testStatistics{i}$ ranks the $(k+1)$-th in $\{\testStatistics{\hypothesisIndex{\text{nc}} \cup \{i\}}\}$, then $H_i$ is rejected iff $k \le \kappa_i$, and $R^{-i}(k) = R^{-i}(0) \mid V_i = 1$.
\end{lemma}

\begin{proof}[Proof of \Cref{lemm:RANC.stable}]

    If $p_i^{-i}(k) < q$, $\hat{\pi}_0^{\hat{\lambda},-i}(0) = \hat{\pi}_0^{\hat{\lambda},-i}(k)$.
    As $k$ increases, $p_{j}^{-i}(k)$ increases, and $R^{-i}(k)$ decreases.
    If $R^{-i}(0) = 0$, then $\kappa_i = -1$ and $H_i$ is never rejected.    
    If $R^{-i}(0) > 0$, then
    \begin{align}\label{proof:eq:RANC.stable}
        \sum_{j \in \hypothesisIndex{}} \1_{\left\{p_j^{-i}(0) \le R^{-i}(0)q/(\No\hat{\pi}_0^{\hat{\lambda},-i}(0)), p_j^{-i}(0) < q \right\}} = R^{-i}(0)
    ~\Longleftrightarrow~
    \sum_{j \in \hypothesisIndex{}} \1_{\{\testStatistics{j} < C_{(\kappa_i+1)}\}} = R^{-i}(0),
    \end{align}
    where $C_{(j)}$ denotes the $j$-th smallest negative control test statistic.
    Given $\testStatistics{\hypothesisIndex{}-\{i\}}$, $\{\testStatistics{\hypothesisIndex{\text{nc}}\cup\{i\}}\}$, 
    \begin{itemize}
        \item [(a)] for $0 \le k \le \kappa_i$, Eq.~\eqref{proof:eq:RANC.stable} holds.
        As a result, $R^{-i}(k) \ge R^{-i}(0)$. Together with $R^{-i}(k) \le R^{-i}(0)$, we have $R^{-i}(k) = R^{-i}(0)$. Since $p_{i}^{-i}(k) = (1+k)/(1+\NoNc) \le (1+\kappa_i)/(1+\NoNc) \le R^{-i}(k)q/(\No\hat{\pi}_0^{\hat{\lambda},-i}(k))$ and $p_{i}^{-i}(k) < q$, then $H_i$ is rejected;
        \item [(b)] for any $k > \kappa_i$, $p_{i}^{-i}(k) \ge q$ or 
        \begin{align*}
            p_{i}^{-i}(k) 
            > \frac{q R^{-i}(0)}{\No\hat{\pi}_0^{\hat{\lambda},-i}(0)} 
            \ge \frac{q R^{-i}(k)}{\No\hat{\pi}_0^{\hat{\lambda},-i}(0)} 
            = \frac{q R^{-i}(k)}{\No\hat{\pi}_0^{\hat{\lambda},-i}(k)},
        \end{align*}
        which implies $H_i$ is not rejected. 
    \end{itemize}    
\end{proof}

\subsection{Convergence rate of the proposed stopping time}

\begin{proof}[Proof of \Cref{prop:convergence.rate}]
    For notation simplicity, we introduce $\bar{F}(t) = 1 - F(t)$, and we write $g(t) = \pi^{t}_{0,\infty} = \bar{F}(t)/(1-t)$, $0 \le t < 1$, and $g(1) = \pi^{1}_{0,\infty} = \overline{\lim}_{t \to 1} \bar{F}(t)/(1-t)$.  
    Since $f$, $f'$ are continuous on $[0,1]$, there exists $f_{\max} > 0$ such that $|f|$, $|f'| \le f_{\max}$ on $[0,1]$. In addition, $g$, $g'$ exist and are continuous on $[0,1]$, and there exists $g_{\max} > 0$ such that $|g|$, $|g'| \le g_{\max}$ on $[0,1]$.
    
    \begin{itemize}
    \item We first prove \Cref{prop:convergence.rate} for the stopping time $\hat{\lambda}_n$ in \eqref{eq:stopping.time}.
    We discuss $q < \lambda^* < 1$, $\lambda^* = q$, and $\lambda^* = 1$, respectively.

    \textbf{Case 1}: $q < \lambda^* < 1$.
    Since $f''$ is continuous on $[0,1]$ and $f'' > 0$ on $[0, 1)$, then by \Cref{lemm:derivative}, $g''(t)$ exists and is continuous, positive on $[0,1)$.
    Therefore, there exists $C_L$, $0 < \eta < (1 - \lambda^*) / 2$ such that $g''(t) \ge C_L$ on $[q, \lambda^* + \eta]$.
    Further, by the mean value theorem, 
     \begin{align}\label{prop:slope.condition}
      g'(t_2)  - g'(t_1) \ge C_L (t_2 - t_1), \quad q \le t_1 < t_2 \le \lambda^* + \eta.
    \end{align}

    Define the events
    \begin{align*}
        \calA_{j+1}^c 
        = \begin{cases}
            \left\{\hat{\pi}_0^{\lambda_j} \le \hat{\pi}_0^{\lambda_{j+1}}\right\}, \quad &  q \le \lambda_j \le \lambda^* - 2 \delta, \\
            \left\{\hat{\pi}_0^{\lambda_j} > \hat{\pi}_0^{\lambda_{j+1}}\right\}, \quad &  \lambda^* + \delta \le \lambda_j < \lambda^* + \eta - \delta. 
        \end{cases}
    \end{align*}
    Notice that 
    \begin{align*}
         \left\{\hat{\pi}_0^{\lambda_j} \le \hat{\pi}_0^{\lambda_{j+1}}\right\} 
         &= \left\{\left(\hat{\pi}_0^{\lambda_{j+1}} - g(\lambda_{j+1})\right) - \left(\hat{\pi}_0^{\lambda_{j}} - g(\lambda_{j})\right) 
        \ge - \left(g(\lambda_{j+1}) - g(\lambda_{j})\right) \right\},\\
        \left\{\hat{\pi}_0^{\lambda_j} > \hat{\pi}_0^{\lambda_{j+1}}\right\}
        &=
        \left\{\left(\hat{\pi}_0^{\lambda_{j+1}} - g(\lambda_{j+1})\right) - \left(\hat{\pi}_0^{\lambda_{j}} - g(\lambda_{j})\right) 
        < - \left(g(\lambda_{j+1}) - g(\lambda_{j})\right) \right\},
    \end{align*}
    and then we analyze $g(\lambda_{j+1}) - g(\lambda_{j})$ 
    and $\left(\hat{\pi}_0^{\lambda_{j+1}} - g(\lambda_{j+1})\right) - \left(\hat{\pi}_0^{\lambda_{j}} - g(\lambda_{j})\right)$ separately.

    Since $g'(t)$ exists and is continuous on $[0,1]$, then $g'(\lambda^*) = 0$.
    For $q \le \lambda_{j} \le \lambda^* - 2 \delta$, 
    \begin{align}\label{proof:eq:mean.bound}
    \begin{split}
        g(\lambda_{j+1}) - g(\lambda_{j})
        &= g'(\lambda_{j} + \gamma \delta) \cdot \delta \quad (\text{for some $\gamma \in [0,1]$ by the mean value theorem})\\
        &=  \left(g'(\lambda_{j} + \gamma \delta) - g'(\lambda^*)\right) \cdot \delta \quad (g'(\lambda^*) = 0) \\
        &\le - C_L(\lambda^* -(\lambda_{j} + \gamma \delta)) \cdot \delta 
        \quad (\text{Condition}~\eqref{prop:slope.condition}) \\
        &\le - C_L(\lambda^* - \lambda_{j+1}) \delta.  
    \end{split}
    \end{align}
    Similarly for $ \lambda^* + \delta \le \lambda_{j} \le \lambda^* + \eta - \delta$,
    \begin{align}\label{proof:eq:mean.bound.2}
         \begin{split}
            g(\lambda_{j+1}) - g(\lambda_{j})
            > C_L(\lambda_{j} - \lambda^*) \delta.  
        \end{split}        
    \end{align}

    Next, for $q \le \lambda_{j} \le \lambda^* + \eta - \delta$,
    \begin{align*}
         &\quad\left(\hat{\pi}_0^{\lambda_{j+1}} - g(\lambda_{j+1})\right) - \left(\hat{\pi}_0^{\lambda_{j}} - g(\lambda_{j})\right) \\
        &= \frac{1+\sum_{i=1}^{\No} \left(\1_{\left\{p_i \ge \lambda_{j+1}\right\}} - \bar{F}(\lambda_{j+1})\right)}{\No(1 - \lambda_{j+1})} - \frac{1 + \left(\1_{\left\{p_i \ge \lambda_{j}\right\}} - \bar{F}(\lambda_{j})\right)}{\No(1 - \lambda_{j})}\\
        &= -\underbrace{\frac{\sum_{i=1}^{\No} \1_{\left\{\lambda_{j} \le p_i < \lambda_{j+1}\right\}} - \left(\bar{F}(\lambda_{j}) - \bar{F}(\lambda_{j+1}) \right)}{\No(1 - \lambda_{j})}}_{:=(\text{I})}
        + \underbrace{\frac{\delta + \delta \sum_{i=1}^{\No} \left(\1_{\left\{p_i \ge \lambda_{j+1}\right\}} - \bar{F}(\lambda_{j+1})\right)}{\No(1 - \lambda_{j})(1 - \lambda_{j+1})}}_{:=(\text{II})}.
    \end{align*}
    For (II), since $\1_{\left\{p_i \ge \lambda_{j+1}\right\}} - \bar{F}(\lambda_{j+1}) \in [-1,1]$ is zero mean, and its variance is upper bounded by $\bar{F}(\lambda_{j+1}) = \bar{F}(\lambda_{j+1}) - \bar{F}(1) \le f_{\max} (1-\lambda_{j+1})$, then by Bernstein's inequality and $1 - (\lambda^* + \eta) \ge (1 - \lambda^*)/2$,
    \begin{align*}
        &\quad~\PP\left(\calB_{j+1}^c := \left\{\left|\frac{\sum_{i=1}^{\No} \1_{\left\{p_i \ge \lambda_{j+1}\right\}} - \bar{F}(\lambda_{j+1})}{\No(1 - \lambda_{j+1})} \right|
        \ge (C_L(1-\lambda^*)/8) \cdot \delta \right\}\right) \\
        &\le 2\exp\left\{-\frac{\left((C_L(1-\lambda^*)/8) (1 - \lambda_{j+1})\No \delta\right)^2/2}{\No f_{\max} (1-\lambda_{j+1}) + (C_L(1-\lambda^*)/8) (1 - \lambda_{j+1})\No \delta/3} \right\} \\
        &
        \le 2\exp\left\{-\frac{C_L^2(1 - \lambda^*)^2 \No\delta^2/128}{f_{\max} + C_L(1-\lambda^*)\delta/24} \right\}.
    \end{align*}
    Further by the union bound,
    \begin{align*}
         \PP\left(\calB := \cap_{q \le \lambda_{j} \le \lambda^* + \eta - \delta} ~\calB_{j+1} \right)
        &\ge 1 - \sum_{q \le \lambda_{j} \le \lambda^* + \eta - \delta} \PP\left(\calB_{j+1}^c \right) \\
        &\ge 1 - \frac{2}{\delta} \exp\left\{-\frac{C_L^2(1 - \lambda^*)^2 \No\delta^2/128}{f_{\max} + C_L(1-\lambda^*)\delta/24} \right\}
        := 1 - p_{\calB^c}.
    \end{align*}
    There exists $N > 0$ such that $p_{\calB^c} = O(\No^{1/3}e^{-\No^{1/3}}) \le M^{-3/2}$ for $\No \ge N$.
    For (I), since $\1_{\left\{\lambda_{j} \le p_i < \lambda_{j+1}\right\}} - \left(\bar{F}(\lambda_{j}) - \bar{F}(\lambda_{j+1}) \right) \in [-1,1]$ is zero mean, and its variance is upper bounded by $\bar{F}(\lambda_{j}) - \bar{F}(\lambda_{j+1}) \le f_{\max} (\lambda_{j+1} - \lambda_{j}) = f_{\max} \delta$. 
    Again by Bernstein's inequality,
    \begin{align*}
        &\PP\left(\calD_{j+1}^c := \left\{|(\text{I})| \ge  C_L(\lambda^* - \lambda_{j+1}) \delta/2 \right \} \right)
        \le 2\exp\left\{-\frac{C_L^2(\lambda^* - \lambda_{j+1})^2(1 - \lambda_{j})^2 \No\delta/8}{f_{\max} + C_L(\lambda^* - \lambda_{j+1})(1 - \lambda_{j})/6} \right\},\quad  \lambda_{j} \le \lambda^* - 2 \delta,\\
        &\PP\left(\calD_{j+1}^c := \left\{|(\text{I})| \ge  C_L(\lambda_j - \lambda^*) \delta/2 \right \} \right)
        \le 2\exp\left\{-\frac{C_L^2(\lambda_j - \lambda^*)^2(1 - \lambda_{j})^2 \No\delta/8}{f_{\max} + C_L(\lambda_j - \lambda^*)(1 - \lambda_{j})/6} \right\}, \quad  \lambda_{j} \ge \lambda^* + \delta.
    \end{align*} 
    Further by the union bound and $1 - (\lambda^* + \eta) > (1-\lambda^*)/2$,
    \begin{align*}
        &\quad~\PP\left(\calD := \left(\cap_{q \le \lambda_{j} \le \lambda^* - 2\delta} ~\calD_{j+1} \right) \cap \left(\cap_{ \lambda^* + \delta \le \lambda_j \le
\lambda^* + \eta-\delta} ~\calD_{j+1} \right) \right)\\
        &\ge 1 - \sum_{q \le \lambda_{j} \le \lambda^* - 2 \delta}
        \exp\left\{-\frac{C_L^2(\lambda^* - \lambda_{j+1})^2(1 - \lambda_{j})^2 \No\delta/8}{f_{\max} + C_L(\lambda^* - \lambda_{j})(1 - \lambda_{j})/6} \right\}\\
        &\quad~- \sum_{\lambda^* + \delta \le \lambda_{j} \le \lambda^* + \eta - \delta}
        \exp\left\{-\frac{C_L^2(\lambda_j - \lambda^*)^2(1 - \lambda_{j})^2 \No\delta/8}{f_{\max} + C_L(\lambda_j - \lambda^*)(1 - \lambda_{j})/6} \right\}\\ 
        &\ge 1 - \frac{1}{\delta}\int_{q}^{\lambda^* + \eta} \exp\left\{-C_L'(1 - \lambda^*)^2\No\delta (\lambda^* - \lambda)^2 \right\} d \lambda \quad (C_L' = C_L^2 / 8 / 4 / (f_{\max} + C_L/6) )\\
        &= 1 - \frac{1}{\delta}\int_{-\eta}^{\lambda^* - q} \exp\left\{-C_L'(1 - \lambda^*)^2\No\delta\lambda^2 \right\} d \lambda \quad (\lambda \leftarrow \lambda^* - \lambda)\\
        &\ge 1 - \frac{1}{\delta}\int_{-\infty}^{\infty} \exp\left\{-C_L'(1 - \lambda^*)^2\No\delta \lambda^2 \right\} d \lambda\\
        &= 1 - \frac{1}{\delta} \sqrt{\frac{\pi }{C_L'(1 - \lambda^*)^2\No\delta }}  
        = 1 - \sqrt{\frac{\pi }{C_L'(1 - \lambda^*)^2 M^3 }}
        := 1 - p_{\calD^c}. \quad \left(\int_{-\infty}^{\infty} e^{-\lambda^2/2\sigma^2} d \lambda = \sqrt{2{\pi\sigma^2}}\right)
    \end{align*}
    Finally, there exists $N > 0$ such that for $\No \ge N$,
    \begin{align*}
        C_L(\lambda^* - \lambda_{j+1}) \delta 
        - \frac{\delta + (C_L(1-\lambda^*)/8) \delta^2\No(1 - \lambda_{j+1})}{\No (1 - \lambda_{j}) (1 - \lambda_{j+1})}
        &\ge C_L (\lambda^* - \lambda_{j+1}) \delta /2,\quad \lambda_{j+1} \le \lambda^* - \delta, \\
        C_L(\lambda_j - \lambda^*) \delta 
        - \frac{\delta + (C_L(1-\lambda^*)/8) \delta^2\No(1 - \lambda_{j+1})}{\No (1 - \lambda_{j}) (1 - \lambda_{j+1})}
        &\ge C_L (\lambda_j - \lambda^*) \delta /2, \quad  \lambda^* + \delta \le \lambda_j \le \lambda^* + 2\delta.
    \end{align*}
    Then $\calB \cap \calD  \subseteq \cap_{q \le \lambda_{j} \le \lambda^* - 2\delta} ~\calA_{j+1} \subseteq \{\hat{\lambda} \ge \lambda^* - \delta\}$ and $\calB \cap \calD \subseteq \calA_{\lceil\lambda^*-q\rceil/\delta+2} \subseteq \{\hat{\lambda} \le \lambda^* + 3 \delta\}$ imply
    \begin{align*}
        \PP\left(\lambda^* - \delta \le \hat{\lambda} \le \lambda^* + 3 \delta \right)
        \ge \PP\left(\calB \cap \calD \right)
        \ge 1 - p_{\calB^c} - p_{\calD^c}.
    \end{align*}
    On $\calB \cap \calD$, we also obtain by $|g'(t)| \le g_{\max}$,
    \begin{align*}
        \hat{\pi}_0^{\hat{\lambda}}
        \le \hat{\pi}_0^{{\lambda}_{\lceil\lambda^*-q\rceil/\delta-1}}
        \le g({\lambda}_{\lceil\lambda^*-q\rceil/\delta-1}) + (C_L(1-\lambda^*)/8) \cdot \delta
        \le g(\lambda^*) + g_{\max} \delta + (C_L(1-\lambda^*)/8) \cdot \delta. 
    \end{align*}

    \textbf{Case 2}: $\lambda^* = q$. Notice that $\calB \cap \calD \subseteq \calA_{2} \subseteq \{\hat{\lambda} \le \lambda^* + 3 \delta\}$ still holds, thus $\PP\left(\hat{\lambda} \le \lambda^* + 3\delta \right) \ge 1 - p_{\calB^c} - p_{\calD^c}$. 
    The analysis of $\hat{\pi}_0^{\hat{\lambda}}$ is similar to that of case 1.

    \textbf{Case 3}: $\lambda^* = 1$. For any $0 < \varepsilon < 1$, 
    \begin{align*}
        \PP(\hat{\lambda} \ge 1 - \varepsilon)
        &\ge \PP\left(\left(\cap_{q \le \lambda_{j} \le (1-\varepsilon)} ~\calB_{j+1} \right) \cap \left(\cap_{q \le \lambda_{j} \le (1-\varepsilon)} ~\calD_{j+1} \right) \right)\\
        &\ge 1 - \frac{2}{\delta} \exp\left\{-\frac{C_L^2(1 - \varepsilon)^2 \No\delta^2/32}{f_{\max} + C_L(1-\varepsilon)\delta/12} \right\}
        - \sqrt{\frac{\pi }{4 C_L'(1 - \varepsilon)^2 M^3}} \\
        &\to 1 - \sqrt{\frac{\pi }{4 C_L'(1 - \varepsilon)^2 M^3}}, \quad n \to \infty.
    \end{align*}
    Let $M = \log(\No) \to \infty$ and we have the consistency of $\hat{\lambda}$.
    On $\left(\cap_{q \le \lambda_{j} \le (1-\varepsilon)} ~\calB_{j+1} \right) \cap \left(\cap_{q \le \lambda_{j} \le (1-\varepsilon)} ~\calD_{j+1} \right)$,
    \begin{align*}
        \hat{\pi}_0^{\hat{\lambda}}
        \le \hat{\pi}_0^{{\lambda}_{\lceil 1 - \varepsilon -q\rceil/\delta-1}} 
        \le g({\lambda}_{\lceil 1 - \varepsilon -q\rceil/\delta-1}) + (C_L(1-\varepsilon)/4) \cdot \delta 
        \le g(1) + g_{\max} (\varepsilon + \delta) + (C_L(1-\varepsilon)/4) \cdot \delta. 
    \end{align*}

    \item We next prove \Cref{prop:convergence.rate} for the robust version stopping time $\hat{\lambda}_n'$ in \eqref{eq:stopping.time.robust}. The proof largely follows that of the non-robust version except that we need to additionally characterize the estimated variance $\hat{V}^{\lambda}$.   
    
    \textbf{Case 1}: $q < \lambda^* < 1$.
    Define the events
    \begin{align*}
        \calA_{j+1}^c 
        = \begin{cases}
            \left\{\hat{\pi}_0^{\lambda_j} + \hat{V}^{\lambda_j} \le \hat{\pi}_0^{\lambda_{j+1}} + \hat{V}^{\lambda_{j+1}}\right\}, \quad &  q \le \lambda_j \le \lambda^* - 2 \delta, \\
            \left\{\hat{\pi}_0^{\lambda_j} + \hat{V}^{\lambda_j} \le \hat{\pi}_0^{\lambda_{j+1}} + \hat{V}^{\lambda_{j+1}}\right\}, \quad &  \lambda^* + \delta \le \lambda_j < \lambda^* + \eta - \delta. 
        \end{cases}
    \end{align*}
    Notice that 
    \begin{align*}
         \left\{\hat{\pi}_0^{\lambda_j} + \hat{V}^{\lambda_j} \le \hat{\pi}_0^{\lambda_{j+1}} + \hat{V}^{\lambda_{j+1}}\right\} 
         &\subseteq \left\{\left(\hat{\pi}_0^{\lambda_{j+1}} - g(\lambda_{j+1})\right) - \left(\hat{\pi}_0^{\lambda_{j}} - g(\lambda_{j})\right) 
        \ge - \left(g(\lambda_{j+1}) - g(\lambda_{j})\right) - \hat{V}^{\lambda_{j+1}} \right\},\\
        \left\{\hat{\pi}_0^{\lambda_j} + \hat{V}^{\lambda_j} \ge \hat{\pi}_0^{\lambda_{j+1}} + \hat{V}^{\lambda_{j+1}}\right\}
        &~\subseteq
        \left\{\left(\hat{\pi}_0^{\lambda_{j+1}} - g(\lambda_{j+1})\right) - \left(\hat{\pi}_0^{\lambda_{j}} - g(\lambda_{j})\right) 
        \le - \left(g(\lambda_{j+1}) - g(\lambda_{j})\right) + \hat{V}^{\lambda_{j}} \right\},
    \end{align*}
    and we only need to analyze $\hat{V}^{\lambda_{j}}$ additionally.

    For $q \le \lambda_{j} \le \lambda^* + \eta$,
    \begin{align}\label{proof:eq:variance.bound}
        \hat{V}^{\lambda_{j}} 
        \le \frac{1}{\No} \cdot \frac{1}{4(1-\lambda_{j})^2}
        \le \frac{1}{\No} \cdot \frac{1}{(1-\lambda^*)^2}.
    \end{align}   
    There exists $N > 0$ such that for $\No \ge N$,
    \begin{align*}
        C_L(\lambda^* - \lambda_{j+1}) \delta -  \frac{1}{\No(1-\lambda^*)^2} - \frac{\delta + (C_L(1-\lambda^*)/8) \delta^2\No(1 - \lambda_{j+1})}{\No (1 - \lambda_{j}) (1 - \lambda_{j+1})}
        &\ge C_L (\lambda^* - \lambda_{j+1}) \delta /2,\quad \lambda_{j+1} \le \lambda^* - \delta, \\
        C_L(\lambda_j - \lambda^*) \delta - \frac{1}{\No(1-\lambda^*)^2} - \frac{\delta + (C_L(1-\lambda^*)/8) \delta^2\No(1 - \lambda_{j+1})}{\No (1 - \lambda_{j}) (1 - \lambda_{j+1})}
        &\ge C_L (\lambda_j - \lambda^*) \delta /2, \quad  \lambda_j \ge \lambda^* + \delta.
    \end{align*}
    Then $\calB \cap \calD  \subseteq \cap_{q \le \lambda_{j} \le \lambda^* - 2\delta} ~\calA_{j+1} \subseteq \{\hat{\lambda}' \ge \lambda^* - \delta\}$ and $\calB \cap \calD \subseteq \calA_{\lceil\lambda^*-q\rceil/\delta+2} \subseteq \{\hat{\lambda}' \le \lambda^* + 3 \delta\}$ imply
    \begin{align*}
        \PP\left(\lambda^* - \delta \le \hat{\lambda}' \le \lambda^* + 3 \delta \right)
        \ge \PP\left(\calB \cap \calD \right)
        \ge 1 - p_{\calB^c} - p_{\calD^c}.
    \end{align*}
    On $\calB \cap \calD$, we also obtain
    \begin{align*}
        \hat{\pi}_0^{\hat{\lambda}'}
        &\le \hat{\pi}_0^{{\lambda}_{\lceil\lambda^*-q\rceil/\delta-1}} + \hat{V}^{{\lambda}_{\lceil\lambda^*-q\rceil/\delta-1}} \\
        &\le g({\lambda}_{\lceil\lambda^*-q\rceil/\delta-1}) + (C_L(1-\lambda^*)/8) \cdot \delta + \frac{1}{\No (1 - \lambda^*)^2}\\
        &\le g(\lambda^*) + g_{\max} \delta + (C_L(1-\lambda^*)/8) \cdot \delta + \frac{1}{\No (1 - \lambda^*)^2}. \quad (|g'(t)| \le g_{\max})
    \end{align*}

    \textbf{Case 2}: $\lambda^* = q$. The same argument as that of $\hat{\lambda}_{\No}$.

    \textbf{Case 3}: $\lambda^* = 1$. For any $0 < \varepsilon < 1$, 
    \begin{align*}
        \PP(\hat{\lambda}' \ge 1 - \varepsilon)
        &\ge \PP\left(\left(\cap_{q \le \lambda_{j} \le (1-\varepsilon)} ~\calB_{j+1} \right) \cap \left(\cap_{q \le \lambda_{j} \le (1-\varepsilon)} ~\calD_{j+1} \right) \right)\\
        &\ge 1 - \frac{2}{\delta} \exp\left\{-\frac{C_L^2(1 - \varepsilon)^2 \No\delta^2/32}{f_{\max} + C_L(1-\varepsilon)\delta/12} \right\}
        - \sqrt{\frac{\pi }{4 C_L'(1 - \varepsilon)^2 M^3}} \\
        &\to 1 - \sqrt{\frac{\pi }{4 C_L'(1 - \varepsilon)^2 M^3}}, \quad n \to \infty.
    \end{align*}
    Let $M \to \infty$ gives the consistency results.
    On $\left(\cap_{q \le \lambda_{j} \le (1-\varepsilon)} ~\calB_{j+1} \right) \cap \left(\cap_{q \le \lambda_{j} \le (1-\varepsilon)} ~\calD_{j+1} \right)$,
    \begin{align*}
        \hat{\pi}_0^{\hat{\lambda}'}
        &\le \hat{\pi}_0^{{\lambda}_{\lceil 1 - \varepsilon -q\rceil/\delta-1}} + \hat{V}^{{\lambda}_{\lceil 1 - \varepsilon -q\rceil/\delta-1}} \\
        &\le g({\lambda}_{\lceil 1 - \varepsilon -q\rceil/\delta-1}) + (C_L(1-\varepsilon)/4) \cdot \delta + \frac{1}{4 \No (1 - \varepsilon)^2}\\
        &\le g(1) + g_{\max} (\varepsilon + \delta) + (C_L(1-\varepsilon)/4) \cdot \delta + \frac{1}{4 \No (1 - \varepsilon)^2}.
    \end{align*}

    \end{itemize}
    
\end{proof}

\begin{lemma}\label{lemm:derivative}
    Suppose for some $m \in \NN$, $\frac{d f^k}{d t^k}(t)$ exists and is continuous on $[0,1]$, $0 \le k \le m$. Then
    \begin{align}\label{eq:lemm:derivative}
        \frac{d g^k}{d t^k}(t)
        = \frac{k!}{(1-t)^{k+1}} \int_t^1 \int_t^{t_1} \cdots \int_t^{t_k}  \frac{d f^k}{d t^k}(t_{k+1}) ~d t_{k+1} \cdots d t_2 d t_1, \quad \forall~0 \le k \le m.
    \end{align}
\end{lemma}

\begin{proof}[Proof of \Cref{lemm:derivative}]
    We prove by induction. For $k = 0$,
    \begin{align*}
        g(t) 
        = \frac{\bar{F}(t)}{1 - t}
        = \frac{\int_t^{1} f(t_1) ~d t_1}{1 - t}.
    \end{align*}
    Assume Eq.~\eqref{eq:lemm:derivative} is true for $k$, then consider the case of $k+1$, 
    \begin{align}\label{proof:eq:integral}
        \begin{split}
            \frac{d g^{k+1}}{d t^{k+1}}(t) 
            &= \frac{d}{dt}\frac{d g^{k}}{d t^k}(t)\\ 
            &= \frac{(k+1)!}{(1-t)^{k+2}} \int_t^1 \int_t^{t_1} \cdots \int_t^{t_k}  \frac{d f^k}{d t^k}(t_{k+1}) ~d t_{k+1} \cdots d t_2 d t_1 \\
            &\quad~+ \frac{k!}{(1-t)^{k+1}} \int_t^1 \int_t^{t_1} \cdots \int_t^{t_{k-1}} \left(-\frac{d f^{k}}{d t^{k}}(t) \right) d t_{k} \cdots d t_2 d t_1.
        \end{split}
    \end{align}
Notice that
\begin{align}\label{proof:eq:integral.2}
    \begin{split}
    &\quad~\frac{k!}{(1-t)^{k+1}} \int_t^1 \int_t^{t_1} \cdots \int_t^{t_{k-1}} \left(-\frac{d f^{k}}{d t^{k}}(t) \right) d t_{k} \cdots d t_2 d t_1 \\
    &= \frac{k!}{(1-t)^{k+1}} \cdot \left(-\frac{d f^{k}}{d t^{k}}(t) \cdot \frac{(1-t)^{k}}{k!}\right) \\
    &= \frac{(k+1)!}{(1-t)^{k+2}} \cdot \left(-\frac{d f^{k}}{d t^{k}}(t) \cdot \frac{(1-t)^{k+1}}{(k+1)!}\right) \\
     &= \frac{(k+1)!}{(1-t)^{k+2}} \int_t^1 \int_t^{t_1} \cdots \int_t^{t_k} \left(-\frac{d f^k}{d t^k}(t) \right) d t_{k+1} \cdots d t_2 d t_1.
 \end{split}
\end{align}
Plug Eq.~\eqref{proof:eq:integral.2} into Eq.~\eqref{proof:eq:integral}, 
\begin{align*}
     \frac{d g^{k+1}}{d t^{k+1}}(t) 
     &= \frac{(k+1)!}{(1-t)^{k+2}} \int_t^1 \int_t^{t_1} \cdots \int_t^{t_k} \frac{d f^k}{d t^k}(t_{k+1}) -\frac{d f^k}{d t^k}(t) ~d t_{k+1} \cdots d t_2 d t_1 \\
     &= \frac{(k+1)!}{(1-t)^{k+2}} \int_t^1 \int_t^{t_1} \cdots \int_t^{t_k} \int_t^{t_{k+1}} \frac{d f^{k+1}}{d t^{k+1}}(t_{k+2}) ~d t_{k+2} d t_{k+1} \cdots d t_2 d t_1,
\end{align*}
and we have finished the proof.
\end{proof}

\end{document}